\documentclass[12pt,a4paper,UKenglish]{article}
\usepackage[margin=1in]{geometry}

\usepackage{amssymb}
\usepackage{latexsym}
\usepackage{amsmath, amsthm, mathtools}
\usepackage{bussproofs}
\usepackage[shortlabels]{enumitem} %

\usepackage{hyperref}

\usepackage{multicol}

\usepackage{authblk}

\usepackage{todonotes}

\newtheorem{theorem}{Theorem}

\newtheorem{lemma}[theorem]{Lemma}
\newtheorem{proposition}[theorem]{Proposition}
\theoremstyle{remark}

\newtheorem*{switching lemma}{Switching Lemma}
\newtheorem*{open problem}{Open Problem}
\theoremstyle{definition}
\newtheorem{definition}[theorem]{Definition}
\newtheorem*{definition*}{Definition}

\DeclareMathOperator{\dom}{dom}
\DeclareMathOperator{\im}{im}

\title{Failure of Feasible Disjunction Property for $k$-DNF Resolution and NP-hardness of Automating It}

\author{Michal Garl\'{\i}k\thanks{Funded by European Research Council (ERC) under the European Union's Horizon 2020 research and innovation programme, grant agreement ERC-2014-CoG 648276 (AUTAR).}
}
\affil{Dept. Ci\`encies de la Computaci\'o\\
Universitat Polit\`ecnica de Catalunya\\
C. Jordi Girona, 1-3\\
08034 Barcelona, Spain\\
email: \texttt{mgarlik@cs.upc.edu}
}

\begin{document}

\maketitle

\begin{abstract}
We show that for every integer $k \geq 2$, the $\textnormal{Res}(k)$ propositional proof system does not have the weak feasible disjunction property. Next, we generalize a recent result of Atserias and M\"uller \cite{atserias-muller2019-focs} to $\textnormal{Res}(k)$. We show that if NP is not included in P (resp. QP, SUBEXP) then for every integer $k \geq 1$, $\textnormal{Res}(k)$ is not automatable in polynomial (resp. quasi-polynomial, subexponential) time.
\end{abstract}

\section{Introduction}

Following Pudl\'ak \cite{Pudlak2003}, a proof system $P$ has \emph{weak feasible disjunction property} if there exists a polynomial $p$ such that if a formula $A \lor B$, in which $A$ and $B$ do not share variables, has a $P$ proof of length $t$, then either $A$ or $B$ has a $P$ proof of length $p(t)$. We deal with refutation systems in this paper, which for the preceding definition amounts to replacing in it `$\lor$' by `$\land$' and `proof' by `refutation'. It is known and easy to see that resolution has the weak feasible disjunction property. Resolution also has feasible interpolation, a prominent concept in proof complexity introduced by Kraj\'{i}\v{c}ek \cite{Krajicek1994, Krajicek1997}. A refutation system $P$ has \emph{feasible interpolation} if there is a polynomial $p$ and an algorithm that when given as input a refutation $\Pi$ of size $r$ of a CNF $A(\overline{x}, \overline{y}) \land B(\overline{x}, \overline{z})$, where $\overline{y}, \overline{x}, \overline{z}$ are disjoint sets of propositional variables, and a truth assignment $\sigma$ to the variables $\overline{x}$ outputs in time $p(r)$ a value $i \in \{0,1\}$ such that if $i=0$ then $A \! \restriction \! \sigma$ is unsatisfiable and if $i=1$ then $B \! \restriction \! \sigma$ is unsatisfiable. Here $F \! \restriction \! \sigma$ denotes the formula obtained from $F$ by an application of a partial truth assignment $\sigma$ to the variables of $F$ that are in the domain of $\sigma$.

Pudl\'ak \cite{Pudlak2003} comments that so far the weak feasible disjunction property has been observed in all proof systems that were shown to have feasible interpolation. This is because known feasible interpolation algorithms, like those in Chapter 17.7 in \cite{Krajicek2019-book}, actually construct a refutation of one of the conjuncts.

A proof system $P$ is \emph{polynomially bounded} if there is a polynomial $p$ such that any tautology of size $r$ has a $P$ proof of size $p(r)$. A fundamental problem in proof complexity is to show that no polynomially bounded proof system exists. This is equivalent to establishing $\textnormal{NP} \neq \textnormal{coNP}$, as observed by Cook and Reckhow \cite{cook-reckhow1979}. There is a potentially useful observation by Kraj\'{i}\v{c}ek \cite{Krajicek2011} that for the purpose of proving that some proof system $P$ is not polynomially bounded we may assume without a loss of generality that $P$ admits the weak feasible disjunction property. This readily follows from the fact that if a disjunction of two formulas that do not share variables is a tautology, then one of the disjuncts is. 

A propositional version of the negation of the \emph{reflection principle} for a proof system $P$ is a conjunction of a propositional formula expressing that `$\overline{z}$ is a satisfying assignment of formula $\overline{x}$ of length $r$' and a propositional formula expressing that `$\overline{y}$ is a $P$ refutation of length $t$ of formula $\overline{x}$ of length $r$'. Here $P, t, r$ are fixed parameters and $\overline{x},\overline{y},\overline{z}$ are disjoint sets of variables. When we plug in for the common variables  $\overline{x}$ some formula $F$ of length $r$, we denote the conjunction by $\textnormal{SAT}^F \land \textnormal{REF}^F_{P,t}$, and we call the second conjunct a $P$ \emph{refutation statement for} $F$. We need to define one very mild requirement on a proof system in order to state a result from \cite{Pudlak2003} about the weak feasible disjunction property that is the main source of motivation for this paper. 
We say that $P$ is \emph{closed under restrictions} if there is a polynomial $p$ such that whenever $F$ has a $P$ proof of length $t$ and $\sigma$ is a partial truth assignment to the variables of $F$, then there is a $P$ proof of $F \! \restriction \! \sigma$ of length at most $p(t)$.

There is a proposition proved in \cite{Pudlak2003} saying that if a proof system $P$ has the weak feasible disjunction property, has polynomial-size proofs of the reflection principle for $P$, is closed under restrictions, and has the property that given a $P$ proof of $\neg \textnormal{SAT}^{\neg F}$ there is at most polynomially longer $P$ proof of $F$, then for every formula $F$ and every integer $t$ which is at least the size of $F$, either there is a $P$ proof of $F$ of length $t^{O(1)}$, or there is a $t^{O(1)}$ long $P$ proof of $\neg \textnormal{REF}^{\neg F}_{P,t}$. Pudl\'ak comments that the conclusion of this proposition seems unlikely (and therefore it seems unlikely that a proof system satisfying the remaining three reasonable properties has the weak feasible disjunction property). He concludes that the weak feasible disjunction property is very unlikely to occur unless the system is very weak. Motivated to find and emphasize the contrast between resolution and $\textnormal{Res}(2)$ (see Section \ref{sec:preliminaries}) in this respect, we show the following theorem.

\begin{theorem}
\label{thm:res(k)-no-wfdp}
For every integer $k \geq 2$, $\textnormal{Res}(k)$ does not have the weak feasible disjunction property. Moreover, there are families $\{A_n\}_{n \geq 1}$ and $\{B_{n,k}\}_{n \geq 1, k \geq 1}$ of CNFs, where $A_n$ has size $n^{O(1)}$, $B_{n,k}$ has size $n^{O(k)}$, and $A_n$ and  $B_{n,k}$ do not share any variables, such that all the following hold:
\begin{enumerate}[(i)]
\item \label{item:lower_bound_on_A_n} There exists $\alpha > 0$ and an integer $n_1$ such that for every $k \geq 1$ and $n \geq n_1$, any $\text{Res}(k)$ refutations of $A_n$ has size greater than $2^{n^{\alpha}}$.
\item \label{item:lower_bound_on_B_n} For every $k \geq 1$ there is $\beta > 0$ and an integer $n_2$ such that for every $n \geq n_2$, any $\textnormal{Res}(k)$ refutation of $B_{n,k}$ has size greater than $2^{\beta n}$. 
\item \label{item:upper_bound_on_A_n_and_B_n} For all integers $n \geq 1$ and $k \geq 1$, $A_n \land B_{n,k}$ has a $\textnormal{Res}(2)$ refutation of size $O(k^2 n^{7k+7})$.
\end{enumerate}
\end{theorem}

The idea is to employ a reflection, but instead of the reflection principle for $\textnormal{Res}(k)$, which would correspond to the hypothesis of Pudl\'ak's proposition above, we work with the reflection principle for resolution and make it harder by the relativization technique of Dantchev and Riis \cite{Dantchev-Riis2003}. More precisely, we replace in the reflection principle the resolution refutation statement by its $k$-fold relativization. Most of this paper is then concerned with proving length lower bounds on $\textnormal{Res}(k)$ refutations of a version of the $k$-fold relativization of $\textnormal{REF}^F_{\textnormal{Res},t}$ for every unsatisfiable CNF $F$ (Theorem \ref{thm:main_size_lb_for_RkREF^F_st}). This lower bound will be used to prove item \ref{item:lower_bound_on_B_n} above, but since it works for every unsatisfiable $F$, item \ref{item:lower_bound_on_A_n} will be easy to get choosing $F$ to be hard enough for $\textnormal{Res}(k)$. The upper bound, item \ref{item:upper_bound_on_A_n_and_B_n}, generalizes upper bounds for similar formulas \cite{Atserias-Bonet2004, atserias-muller2019-focs, Garlik2019-mfcs}, which all build on an idea from \cite{Pudlak2003}.

To prove Theorem \ref{thm:main_size_lb_for_RkREF^F_st}, the mentioned main lower bound, we develop a switching lemma in the spirit of \cite{segerlind-buss-impagliazzo} but respecting the functional properties of the formula $\textnormal{REF}^F_{\textnormal{Res},t}$. This will come at a cost of worse parameters in the switching lemma and its narrowed applicability in terms of random restrictions it works for.

Our second result is a generalization of conditional non-automatability results for resolution \cite{atserias-muller2019-focs} to the systems $\textnormal{Res}(k)$. Following \cite{Bonet-Pitassi-Raz,Atserias-Bonet2004} and \cite{atserias-muller2019-focs}, we say that a refutation system $P$ is \emph{automatable in time} $T: \mathbb{N} \rightarrow \mathbb{N}$ if there is an algorithm that when given as input an unsatisfiable CNF $F$ of size $r$ outputs a $P$ refutation of $F$ in time $T(r + s_P(F))$, where $s_P(F)$ is the length of a shortest $P$ refutation of $F$. If the function $T$ is a polynomial, then $P$ is simply called \emph{automatable}. A refutation system $P$ is \emph{weakly automatable} if there is a refutation system $Q$, a polynomial $p$, and an algorithm that when given as input an unsatisfiable CNF $F$ of size $r$ outputs a $Q$ refutation of $F$ in time $p(r+s_P(F))$. It is known that feasible interpolation is implied by weak automatability in refutation systems that are closed under restrictions (see Theorem 3 in \cite{Atserias-Bonet2004}).

First negative automatability results were obtained by Kraj\'{i}\v{c}ek and Pudl\'{a}k \cite{Krajicek-Pudlak1998} who showed that Extended Frege systems  do not have feasible interpolation assuming that RSA is secure against P/poly. Bonet et al. \cite{Bonet-Pitassi-Raz,Bonet-Domingo-Gavalda-Maciel-Pitassi-Raz2004} showed that Frege systems and constant-depth Frege systems do not have feasible interpolation assuming the Diffie-Hellman key exchange procedure is secure against polynomial and subexponential size circuits, respectively. All these proof systems are closed under restrictions, hence these results conditionally rule out weak automatability and automatability. 
As for resolution, before a recent breakthrough by Atserias and M\"uller \cite{atserias-muller2019-focs} who showed that resolution is not automatable unless P = NP, it was known by a result of Alekhnovich and Razborov \cite{Alekhnovich-Razborov2008} that resolution is not automatable unless W[P] = FPT. Here W[P] is the class of parametrized problems that are fixed-parameter reducible to the problem of deciding if a monotone circuit $C$ has a satisfying assignment of Hamming weight $k$. We refer an interested reader to the introduction section of \cite{atserias-muller2019-focs} for more on the history of the automatability problem.

Let QP denote the class of problems decidable in quasi-polynomial time $2^{(\log n)^{O(1)}}$, and let SUBEXP denote the class of problems decidable in subexponential time $2^{n^{o(1)}}$. We show the following theorem, which was proved for $k=1$ in \cite{atserias-muller2019-focs}.

\begin{theorem}
\phantomsection\label{thm:res(k)-not_automatable} %
\begin{enumerate}[1.]
\item If $\textnormal{NP} \not \subseteq \textnormal{P}$ then for every integer $k \geq 1$, $\textnormal{Res}(k)$ is not automatable in polynomial time.
\item If $\textnormal{NP} \not \subseteq \textnormal{QP}$ then for every integer $k \geq 1$, $\textnormal{Res}(k)$ is not automatable in quasi-polynomial time.
\item If $\textnormal{NP} \not \subseteq \textnormal{SUBEXP}$ then for every integer $k \geq 1$, $\textnormal{Res}(k)$ is not automatable in subexponential time.
\end{enumerate}
\end{theorem}

The basic idea of the proof is the same as in \cite{atserias-muller2019-focs}: to map every formula $F$ to a resolution refutation statement for $F$, and show that if $F$ is satisfiable then the refutation statement has a polynomial-length $\textnormal{Res}(k)$ refutation, and if $F$ is unsatisfiable then the refutation statement requires long $\textnormal{Res}(k)$ refutations. An automating algorithm that finds short refutations quickly enough can then be used to distinguishing between the two situations, and hence to solve SAT. We thus need to show strong lower bounds on the length of $\textnormal{Res}(k)$ refutations of a version of resolution refutation statements. For this we use the already discussed Theorem \ref{thm:main_size_lb_for_RkREF^F_st} once more.

\section{Preliminaries}
\label{sec:preliminaries}
For an integer $s$, the set $\{1, \ldots, s\}$ is denoted by $[s]$. 
We write $\dom(\sigma), \im(\sigma)$ for the domain and image of a function $\sigma$.
Two functions $\sigma, \tau$ are \emph{compatible} if $\sigma \cup \tau$ is a function.
If $x$ is a propositional variable, the \emph{positive literal of} $x$, denoted by $x^1$, is $x$, and the \emph{negative literal of} $x$, denoted by $x^0$, is $\neg x$. A \emph{clause} is a set of literals. A clause is written as a disjunction of its elements. A \emph{term} is a set of literals, and is written as a conjunction of the literals. A \emph{CNF} is a set of clauses, written as a conjunction of the clauses. A $k$-\emph{CNF} is a CNF whose every clause has at most $k$ literals. A \emph{DNF} is a set of terms, written as a disjunction of the terms. A $k$-\emph{DNF} is a DNF whose every term has at most $k$ literals. We will identify 1-DNFs with clauses. 
A clause is \emph{non-tautological} if it does not contain both the positive and negative literal of the same variable. A clause $C$ is a \emph{weakening} of a clause $D$ if $D \subseteq C$. A clause $D$ is the \emph{resolvent of} clauses $C_1$ and $C_2$ \emph{on} a variable $x$ if $x \in C_1, \neg x \in C_2$ and $D = (C_1\setminus \{x\}) \cup (C_2 \setminus \{\neg x\})$. If $E$ is a weakening of the resolvent of $C_1$ and $C_2$ on $x$, we say that $E$ is obtained by the \emph{resolution rule} from $C_1$ and $C_2$, and we call $C_1$ and $C_2$ the \emph{premises} of the rule.

Let $F$ be a CNF and $C$ a clause. A \emph{resolution derivation of} $C$ \emph{from} $F$ is a sequence of clauses $\Pi = (C_1, \ldots , C_s)$ such that $C_s = C$ and for all $u \in [s]$, $C_u$ is a weakening of a clause in $F$, or there are $v,w \in [u-1]$ such that $C_u$ is obtained by the resolution rule from $C_v$ and $C_w$. 
A \emph{resolution refutation of} $F$ is a resolution derivation of the empty clause from $F$.
The \emph{length} of a resolution derivation $\Pi = (C_1, \ldots , C_s)$ is $s$. 
For $u \in [s]$, the \emph{height of} $u$ \emph{in} $\Pi$ is the maximum $h$ such that there is a subsequence $(C_{u_1}, \ldots , C_{u_h})$ of $\Pi$ in which $u_h = u$ and for each $i \in [h-1]$, $C_{u_i}$ is a premise of a resolution rule by which $C_{u_{i+1}}$ is obtained in $\Pi$. 
The \emph{height} of $\Pi$ is the maximum height of $u$ in $\Pi$ for $u \in [s]$. 

A \emph{partial assignment} to the variables $x_1, \ldots , x_n$ is a partial map from $\{ x_1, \ldots , x_n \}$ to $\{0,1\}$. Let $\sigma$ be a partial assignment. The CNF $F \! \restriction \! \sigma$ is formed from $F$ by removing every clause containing a literal satisfied by $\sigma$, and removing every literal falsified by $\sigma$ from the remaining clauses. If $\Pi = (C_1, \ldots , C_s)$ is a sequence of clauses, $\Pi \!\restriction \! \sigma$ is formed from $\Pi$ by the same operations. Note that if $\Pi$ is a resolution refutation of $F$, then $\Pi \!\restriction \! \sigma$ is a resolution refutation of $F \!\restriction \! \sigma$. 

The $\text{Res}(k)$ refutation system is a generalization of resolution introduced by Kraj\'{i}\v{c}ek \cite{Krajicek2001}\footnote{In \cite{Krajicek2001} (see also Chapter 5.7 in \cite{Krajicek2019-book}) more general fragments $\textnormal{R}(f)$ of DNF-resolution are introduced, where $f:\mathbb{N} \rightarrow \mathbb{N}$ is non-decreasing and a refutation $\Pi$ is said to have $\textnormal{R}(f)$-\emph{size} $s$ if its lines are $f(s)$-DNFs and $|\Pi| \leq s$. In the present paper we work with constant functions $f$.}. 
Its lines are $k$-DNFs and it has the following inference rules ($A,B$ are $k$-DNFs, $j \in [k]$, and $l,l_1,\ldots, l_j$ are literals):
\begingroup
\setlength{\tabcolsep}{15pt} %
\renewcommand{\arraystretch}{2.5}  %
\begin{center}
\begin{tabular}{c c}
\AxiomC{$A \lor l_1$}
\AxiomC{$B \lor (l_2 \land \cdots \land l_j)$}
\RightLabel{\,$\land$-introduction}
\BinaryInfC{$A \lor B \lor (l_1 \land \cdots \land l_j)$}
\DisplayProof 
&
\AxiomC{\phantom{A}}
\RightLabel{\,Axiom}
\UnaryInfC{$x \lor \neg x$}
\DisplayProof 
\\
\AxiomC{$A \lor (l_1 \land \cdots \land l_j)$}
\AxiomC{$B \lor \neg l_1 \lor \cdots \lor \neg l_j$}
\RightLabel{\,Cut}
\BinaryInfC{$A \lor B$}
\DisplayProof 
& 
\AxiomC{$A$}
\RightLabel{\,Weakening}
\UnaryInfC{$A \lor B$}
\DisplayProof
\\
\end{tabular}
\end{center}
\endgroup
\noindent Let $F$ be a CNF. A $\textit{Res}(k)$ \emph{derivation from} $F$ is a sequence of $k$-DNFs ($D_1,\ldots, D_s$) so that each $D_i$ either belongs to $F$ or follows from the preceding lines by an application of one of the inference rules. A $\textit{Res}(k)$ \emph{refutation of} $F$ is a $\text{Res}(k)$ derivation from $F$ whose final line is the empty clause. The \emph{length} of a $\textnormal{Res}(k)$ derivation $\Pi = (D_1,\ldots, D_s)$, denoted by $|\Pi|$, is $s$. The \emph{size} of $\Pi$, denoted by $\textnormal{size}(\Pi)$, is the number of symbols in it.

\section{Resolution Refutations of $\textit{s}$ Levels of $\textit{t}$ Clauses}
\label{sec:res-refut-of-s-t-definition}
Like in \cite{Garlik2019-mfcs}, it will be convenient to work with a variant of resolution in  which the clauses forming a refutation are arranged in layers. All the definitions in this section are taken from \cite{Garlik2019-mfcs}.

\begin{definition}
 Let $F$ be a CNF of $r$ clauses in $n$ variables $x_1, \ldots, x_n$.  We say that $F$ has a \emph{resolution refutation of $s$ levels of $t$ clauses} if there is a sequence of clauses $C_{i,j}$ indexed by all pairs $(i,j) \in [s] \times [t]$, such that each clause $C_{1,j}$ on the first level is a weakening of a clause in $F$, each clause $C_{i,j}$ on level $i \! \in \! [s] \! \setminus \! \{1\}$ is a weakening of the resolvent of two clauses from level $i-1$ on a variable, and the clause $C_{s,t}$ is empty. 
\end{definition}

The following proposition says that insisting that the clauses are arranged in layers is not a very limiting requirement since this system quadratically simulates resolution and preserves the refutation height.

\begin{proposition}[\cite{Garlik2019-mfcs}]
\label{proposition:simulation-of-Res}
If a $(n-1)$-CNF $F$ in $n$ variables has a resolution refutation of height $h$ and length $s$, then $F$ has a resolution refutation of $h$ levels of $3s$ clauses. 
\end{proposition}

We now formalize refutation statements for this system in the same way as in \cite{Garlik2019-mfcs}. 
Let $n,r,s,t$ be integers. Let $F$ be a CNF consisting of $r$ clauses $C_1, \ldots , C_r$ in $n$ variables $x_1, \ldots , x_n$. We define a propositional formula $\textnormal{REF}^{F}_{s,t}$ expressing that $F$ has a resolution refutation of $s$ levels of $t$ clauses.

We first list the variables of $\textnormal{REF}^{F}_{s,t}$. $D$-\emph{variables} $D(i,j,\ell,b)$, $i \in [s]$, $j \in [t], \ell \in [n], b \in \{0,1\}$, encode clauses $C_{i,j}$ as follows: $D(i,j,\ell,1)$ (resp. $D(i,j,\ell,0)$) means that the literal $x_\ell$ (resp. $\neg x_\ell$) is in $C_{i,j}$.  
$L$-\emph{variables} $L(i,j,j')$ (resp. $R$-\emph{variables} $R(i,j,j')$), $i \! \in \! [s] \! \setminus \! \{1\}, j,j' \in [t]$, say that $C_{i-1,j'}$ is a premise of the resolution rule by which $C_{i,j}$ is obtained, and it is the premise containing the positive (resp. negative) literal of the resolved variable. 
$V$-\emph{variables} $V(i,j,\ell)$, $i \! \in \! [s] \! \setminus \! \{1\}, j \in [t], \ell \in [n]$, say that $C_{i,j}$ is obtained by resolving on $x_\ell$.
$I$-\emph{variables} $I(j,m)$, $j \in [t], m \in [r]$, say that $C_{1,j}$ is a weakening of $C_m$.

$\textnormal{REF}^{F}_{s,t}$ is the union of the following fifteen sets of clauses:
\begin{alignat}{2}
\label{refclause:axioms} & \neg I(j,m) \lor D(1,j,\ell,b)   & j \! \in \! [t], m \! \in \! [r], b \! \in \! \{0,1\}, x_\ell^b \! \in \! C_m,
\intertext{clause $C_{1,j}$  contains the literals of $C_m$ assigned to it by $I(j,m)$,}
\label{refclause:nontaut} & \neg D(i,j,\ell,1) \lor \neg D(i,j,\ell,0)  & i \! \in \! [s], j \! \in \! [t], \ell \! \in \! [n],
\intertext{no clause $C_{i,j}$ contains $x_\ell$ and $\neg x_\ell$ at the same time,}
\label{refclause:res-L-cut} & \neg L(i,j,j') \lor \neg V(i,j,\ell) \lor D(i-1,j',\ell,1) & i \! \in \! [s] \! \setminus \! \{1\}, j,j' \! \in \! [t], \ell \! \in \! [n], 
\\
\label{refclause:res-R-cut} & \neg R(i,j,j') \lor \neg V(i,j,\ell) \lor D(i-1,j',\ell,0) & i \! \in \! [s] \! \setminus \! \{1\}, j,j' \! \in \! [t], \ell \! \in \! [n],
\intertext{clause $C_{i-1,j'}$ used as the premise given by $L(i,j,j')$ (resp. $R(i,j,j')$) in resolving on $x_\ell$ must contain $x_\ell$ (resp. $\neg x_\ell$),}
\begin{split}
\label{refclause:res-L-transf} \mathrlap{ \neg L(i,j,j') \lor \neg V(i,j,\ell) \lor \neg D(i-1,j',\ell',b) \lor D(i,j,\ell',b)} \\
\mathrlap{ \, \, \; \; \; \qquad \qquad \qquad \qquad  i \! \in \! [s] \! \setminus \! \{1\}, j,j' \! \in \! [t], \ell,\ell' \! \in \! [n], b \! \in \! \{0,1\}, (\ell',b) \neq (\ell,1),} 
\end{split}
\\[1ex]
\begin{split}
\label{refclause:res-R-transf} \mathrlap{ \neg R(i,j,j') \lor \neg V(i,j,\ell) \lor \neg D(i-1,j',\ell',b) \lor D(i,j,\ell',b)} \\
 \mathrlap{ \, \, \; \; \; \qquad \qquad \qquad \qquad  i \! \in \! [s] \! \setminus \! \{1\}, j,j' \! \in \! [t], \ell,\ell' \! \in \! [n], b \! \in \! \{0,1\}, (\ell',b) \neq (\ell,0),}
\end{split}
\intertext{clause $C_{i,j}$ derived by resolving on $x_\ell$ must contain each literal different from $x_\ell$ (resp. $\neg x_\ell$) from the premise given by $L(i,j,j')$ (resp. $R(i,j,j')$),} 
\label{refclause:empty-clause} & \neg D(s,t,\ell,b) & \ell \! \in \! [n], b \! \in \! \{0,1\}, 
\intertext{clause $C_{s,t}$ is empty,}
\label{refclause:V-dom} & V(i,j,1) \lor V(i,j,2) \lor \ldots \lor V(i,j,n)  & i \! \in \! [s] \! \setminus \! \{1\}, j\! \in \! [t], \\
\label{refclause:I-dom} & I(j,1) \lor I(j,2) \lor \ldots \lor I(j,r)  & j\! \in \! [t], \\
\label{refclause:L-dom} & L(i,j,1) \lor L(i,j,2) \lor \ldots \lor L(i,j,t)  & i \! \in \! [s] \! \setminus \! \{1\}, j\! \in \! [t], \\
\label{refclause:R-dom} & R(i,j,1) \lor R(i,j,2) \lor \ldots \lor R(i,j,t)  & i \! \in \! [s] \! \setminus \! \{1\}, j\! \in \! [t], \\
\label{refclause:V-func} & \neg V(i,j,\ell) \lor \neg V(i,j,\ell')  & i \! \in \! [s] \! \setminus \! \{1\}, j\! \in \! [t], \ell,\ell' \! \in \! [n], \ell \neq \ell',\\
\label{refclause:I-func} & \neg I(j,m) \lor \neg I(j,m')  &  j\! \in \! [t], m,m' \! \in \! [r], m \neq m',\\
\label{refclause:L-func} & \neg L(i,j,j') \lor \neg L(i,j,j'')  & i \! \in \! [s] \! \setminus \! \{1\}, j,j',j'' \! \in \! [t], j' \neq j'',\\
\label{refclause:R-func} & \neg R(i,j,j') \lor \neg R(i,j,j'')  & i \! \in \! [s] \! \setminus \! \{1\}, j,j',j'' \! \in \! [t], j' \neq j'', 
\end{alignat}
the $V,I,L,R$-variables define functions with the required domains and ranges.

\begin{definition}
\label{def:home-pair-and-set-to}
For $i \in [s], j,j' \in [t], \ell \in [n], b \in \{0,1\}, m \in [r]$, we say that $(i,j)$ is the \emph{home pair} of the variable $D(i,j,\ell,b)$, of the variables $R(i,j,j')$, $L(i,j,j')$, $V(i,j,\ell)$ if $i \neq 1$, and of the variable $I(j,m)$ if $i=1$.

We write $V(i,j,\cdot)$ to stand for the set $\{V(i,j,\ell) : \ell \in [n] \}$. Similarly, we write $I(j,\cdot), L(i,j,\cdot)$, and $R(i,j,\cdot)$ to stand for the set $\{I(j,m) : m \in [r]\}$, $\{L(i,j,j') : j' \in [t]\}$, and $\{R(i,j,j') : j' \in [t]\}$, respectively. We denote by $D(i,j,\cdot,\cdot)$ the set $\{D(i,j,\ell,b) : \ell \in [n], b \in \{0,1\}\}$. 

Let $\sigma$ be a partial assignment. We say that $V(i,j,\cdot)$ is \emph{set to} $\ell$ \emph{by} $\sigma$ if $\sigma(V(i,j,\ell)) = 1$ and for all $\ell' \in [n] \! \setminus \{\ell\}$, $\sigma(V(i,j,\ell')) = 0$ . Similarly, we say that $I(j,\cdot)$ is \emph{set to} $m$ \emph{by} $\sigma$ if $\sigma(I(j,m)) = 1$ and for all $m' \in [r]\! \setminus \! \{m\}$ we have $\sigma(I(j,m')) = 0$.
We say that $L(i,j,\cdot)$ (resp. $R(i,j,\cdot)$)   is \emph{set to} $j'$ \emph{by} $\sigma$ if $\sigma(L(i,j,j')) = 1$ (resp. $\sigma(R(i,j,j')) = 1$) and for all $j'' \in [t] \! \setminus \! \{j'\}$,  we have $\sigma(L(i,j,j'')) = 0$ (resp. $\sigma(R(i,j,j'')) = 0$). 
We say that $D(i,j,\cdot, \cdot)$ is \emph{set to} a clause $C_{i,j}$ \emph{by} $\sigma$ if for all $\ell \in [n], b \in \{0,1\}$ we have $\sigma(D(i,j,\ell,b)) = 1$ if $x_\ell^b \in C_{i,j}$ and $\sigma(D(i,j,\ell,b)) = 0$ if $x_\ell^b \not \in C_{i,j}$. 

For $Y \in \{D(i,j,\cdot,\cdot), V(i,j,\cdot), I(j,\cdot), R(i,j,\cdot), L(i,j,\cdot)\}$, we say that $Y$ is \emph{set by} $\sigma$ if $Y$ is set to $v$ by $\sigma$ for some value $v$. We will often omit saying ``by $\sigma$'' if $\sigma$ is clear from the context.
\end{definition}

\section{Reflection Principle for Resolution}
\label{sec:refl-princ-for-res}
We repeat the formulation of a version of the reflection principle from \cite{Garlik2019-mfcs}.
We express the negation of the reflection principle for resolution by a CNF in the form of a conjunction $\textnormal{SAT}^{n,r} \land \textnormal{REF}^{n,r}_{s,t}$. The only shared variables by the formulas $\textnormal{SAT}^{n,r}$ and $\textnormal{REF}^{n,r}_{s,t}$ encode a CNF with $r$ clauses in $n$ variables. The meaning of $\textnormal{SAT}^{n,r}$ is that the encoded CNF is satisfiable, while the meaning of $\textnormal{REF}^{n,r}_{s,t}$ is that the same CNF has a resolution refutation of $s$ levels of $t$ clauses. A formal definition is given next.

Formula $\textnormal{SAT}^{n,r}$ has the following variables. $C$-\emph{variables} $C(m,\ell,b)$, $m \in [r], \ell \in [n], b \in \{0,1\},$ encode clauses $C_m$ as follows: $C(m,\ell,1)$ (resp. $C(m,\ell,0)$) means that the literal $x_\ell$ (resp. $\neg x_\ell$) is in $C_m$. 
$T$-\emph{variables} $T(\ell)$, $\ell \in [n]$, and $T(m,\ell,b)$, $m \in [r], \ell \in [n], b \in \{0,1\}$, encode that an assignment to variables $x_1,\ldots, x_n$ satisfies the CNF $\{C_1,\ldots, C_r\}$. The meaning of $T(\ell)$ is that the literal $x_\ell$ is satisfied by the assignment. The meaning of $T(m,\ell,1)$ (resp. $T(m,\ell,0)$) is that clause $C_m$ is satisfied through the literal $x_\ell$ (resp. $\neg x_\ell$).

We list the clauses of $\textnormal{SAT}^{n,r}$:
\begin{align}
\label{satclause:at-least-one-lit-sat} & T(m,1,1) \lor T(m,1,0) \lor \ldots \lor T(m,n,1) \lor T(m,n,0)& m \in [r],\\
\label{satclause:sat-by-positive} & \neg T(m,\ell,1) \lor T(\ell) & m \in [r], \ell \in [n],\\
\label{satclause:sat-by-negative} & \neg T(m,\ell,0) \lor \neg T(\ell) & m \in [r], \ell \in [n],\\
\label{satclause:sat-lit-in-clause} & \neg T(m,\ell,b) \lor C(m,\ell,b) & m \in [r], \ell \in [n], b \in \{0,1\}.
\end{align}
The meaning of \eqref{satclause:at-least-one-lit-sat} is that clause $C_m$ is satisfied through at least one literal.
Clauses \eqref{satclause:sat-by-positive} and \eqref{satclause:sat-by-negative} say that if $C_m$ is satisfied through a literal, then the literal is satisfied.
The meaning of \eqref{satclause:sat-lit-in-clause} is that if $C_m$ is satisfied through a literal, then it contains the literal. 

Variables of $\textnormal{REF}^{n,r}_{s,t}$ are the variables $C(m,\ell,b)$ of $\textnormal{SAT}^{n,r}$ together with all the variables of $\textnormal{REF}^F_{s,t}$ for some (and every) $F$ of $r$ clauses in $n$ variables. That is, $\textnormal{REF}^{n,r}_{s,t}$ has the following variables:
$C(m,\ell,b)$ for $m \in [r], \ell \in [n], b \in \{0,1\}$;
$D(i,j,\ell,b)$ for $i \in [s], j \in [t], \ell \in [n], b \in \{0,1\}$; 
$R(i,j,j')$ and $L(i,j,j')$ for $i \! \in \! [s] \! \setminus \! \{1\}, j,j' \in [t]$; $V(i,j,\ell)$ for $i \! \in \! [s] \! \setminus \! \{1\}, j \in [t], \ell \in [n]$; $I(j,m)$ for $j \in [t], m \in [r]$.

The clauses of $\textnormal{REF}^{n,r}_{s,t}$ are 
\eqref{refclause:nontaut} - \eqref{refclause:R-func} of  $\textnormal{REF}^F_{s,t}$ together with the following clauses (to replace clauses \eqref{refclause:axioms}): 
\begin{align}
\label{refclause-n-r:axioms}
& \neg I(j,m) \lor \neg C(m,\ell,b) \lor D(1,j,\ell,b) &  j \in [t], m \in [r], \ell \in [n], b \in \{0,1\},
\end{align}
saying that if clause $C_{1,j}$ is a weakening of clause $C_m$, then the former contains each literal of the latter. So the difference from \eqref{refclause:axioms} is that $C_m$ is no longer a clause of some fixed formula $F$, but it is described by $C$-variables.

\begin{proposition}
\label{propos:subst-to-SAT}
Let $F$ be a CNF with $r$ clauses $C_1,\ldots, C_r$ in $n$ variables $x_1,\ldots, x_n$, and let $\gamma_F$ be an assignment such that its domain are all $C$-variables and $\gamma_F (C(m,\ell,b)) = 1$ if $x_\ell^b \in C_m$ and $\gamma_F (C(m,\ell,b)) = 0$ if $x_\ell^b \notin C_m$. There is a substitution $\tau$ that maps the variables of $\textnormal{SAT}^{n,r} \! \restriction \! \gamma_F$ to $\{0,1\} \cup \{ x_\ell^b : \ell \in [n], b \in \{0,1\} \}$ such that $( \textnormal{SAT}^{n,r} \! \restriction \! \gamma_F ) \! \restriction \! \tau$ is $F$ together with some tautological clauses in the variables $x_1,\ldots, x_n$. 
\end{proposition}

\begin{proof}
Define $\tau$ as follows. If $\gamma_F(C(m,\ell,b)) = 0$, then $\tau(T(m,\ell,b)) = 0$. This deletes $T(m,\ell,b)$ from \eqref{satclause:at-least-one-lit-sat} and satisfies \eqref{satclause:sat-lit-in-clause} together with either \eqref{satclause:sat-by-positive} (if $b=1$) or \eqref{satclause:sat-by-negative} (if $b=0$). 
If $\gamma_F(C(m,\ell,b)) = 1$, then \eqref{satclause:sat-lit-in-clause} has been satisfied and we define $\tau(T(m,\ell,b)) = x_\ell^b$ and $\tau(T(\ell)) = x_\ell$. This choice turns \eqref{satclause:sat-by-positive} (if $b=1$) or \eqref{satclause:sat-by-negative} (if $b=0$) into a tautological clause and correctly substitutes the remaining literals of \eqref{satclause:at-least-one-lit-sat} to yield the clause $C_m$ of $F$.
\end{proof}

\section{The Upper Bounds}
\label{sec:refl-princ-upper-bound}

In this section we work with a stronger formulation of the negation of the reflection principle for resolution, expressed by a CNF formula $\textnormal{SAT}^{n,r} \land \textnormal{R}^k\textnormal{REF}^{n,r}_{s,t}$. The difference from the previous formulation $\textnormal{SAT}^{n,r} \land \textnormal{REF}^{n,r}_{s,t}$ is that we have replaced $\textnormal{REF}^{n,r}_{s,t}$ by its $k$-fold relativization $\textnormal{R}^k\textnormal{REF}^{n,r}_{s,t}$. The first-order logic notion of relativization of a first-order formula to a relation was put to use in propositional proof complexity by Dantchev and Riis \cite{Dantchev-Riis2003}. 

We first describe the $k$-fold relativization of $\textnormal{REF}^{F}_{s,t}$, denoted by $\textnormal{R}^k\textnormal{REF}^{F}_{s,t}$. The variables of this CNF are those of $\textnormal{REF}^{F}_{s,t}$ together with new variables $S_u(i,j)$, $(i,j) \in [s] \times [t]$, $u \in [k]$. The meaning of $\textnormal{R}^k\textnormal{REF}^{F}_{s,t}$ is that those clauses $C_{i,j}$ (described by $D$-variables) for which $\bigwedge_{u \in [k]}S_u(i,j)$ is satisfied form a resolution refutation of $F$ of $s$ levels of at most $t$ clauses. 
That is, only the selected clauses $C_{i,j}$ have to form a refutation, and nothing is asked of the clauses that are not selected. 
Formally, $\textnormal{R}^k\textnormal{REF}^{F}_{s,t}$ is the union of the following sets of clauses:
\begingroup
\allowdisplaybreaks %
\begin{alignat}{2}
\label{rrefclause:axioms} & \bigvee_{u\in [k]} \neg S_u(1,j) \lor \neg I(j,m) \lor D(1,j,\ell,b)   & j \! \in \! [t], m \! \in \! [r], b \! \in \! \{0,1\}, x_\ell^b \! \in \! C_m, \\ 
\label{rrefclause:nontaut} & \bigvee_{u\in [k]} \neg S_u(i,j)\lor \neg D(i,j,\ell,1) \lor \neg D(i,j,\ell,0)  & i \! \in \! [s], j \! \in \! [t], \ell \! \in \! [n], \\
\begin{split}
\label{rrefclause:res-L-cut} & \mathrlap{ \bigvee_{u\in [k]} \neg S_u(i,j) \lor \neg L(i,j,j') \lor \neg V(i,j,\ell) \lor D(i-1,j',\ell,1) } \\
\mathrlap{\, \, \, \, \, \, \, \, \; \qquad \qquad \qquad \qquad \qquad \qquad \qquad \qquad \qquad \qquad i \! \in \! [s] \!\setminus \! \{1\}, j,j' \! \in \! [t], \ell \! \in \! [n],} 
\end{split} \\[1ex]
\begin{split}
\label{rrefclause:res-R-cut} & \mathrlap{ \bigvee_{u\in [k]} \neg S_u(i,j) \lor \neg R(i,j,j') \lor \neg V(i,j,\ell) \lor D(i-1,j',\ell,0) } \\
\mathrlap{\, \, \, \, \, \, \, \, \; \qquad \qquad \qquad \qquad \qquad \qquad \qquad \qquad \qquad \qquad i \! \in \! [s] \!\setminus \! \{1\}, j,j' \! \in \! [t], \ell \! \in \! [n],}
\end{split} \\[1ex]
\begin{split}
\label{rrefclause:res-L-transf} \mathrlap{ \bigvee_{u\in [k]} \neg S_u(i,j) \lor \neg L(i,j,j') \lor \neg V(i,j,\ell) \lor \neg D(i-1,j',\ell',b) \lor D(i,j,\ell',b)} \\
\mathrlap{  \, \, \, \, \, \, \, \; \; \; \quad \quad \qquad \qquad \qquad  i \! \in \! [s] \! \setminus \! \{1\}, j,j' \! \in \! [t], \ell,\ell' \! \in \! [n], b \! \in \! \{0,1\}, (\ell',b) \neq (\ell,1),} 
\end{split} \\[1ex]
\begin{split}
\label{rrefclause:res-R-transf} \mathrlap{ \bigvee_{u\in [k]} \neg S_u(i,j) \lor \neg R(i,j,j') \lor \neg V(i,j,\ell) \lor \neg D(i-1,j',\ell',b) \lor D(i,j,\ell',b)} \\
 \mathrlap{  \, \, \, \, \, \, \, \; \; \; \quad \quad \qquad \qquad \qquad  i \! \in \! [s] \! \setminus \! \{1\}, j,j' \! \in \! [t], \ell,\ell' \! \in \! [n], b \! \in \! \{0,1\}, (\ell',b) \neq (\ell,0),}
\end{split} \\[1ex]
\label{rrefclause:empty-clause} & \bigvee_{u \in [k]} \neg S_u(s,t) \lor \neg D(s,t,\ell,b) & \ell \! \in \! [n], b \! \in \! \{0,1\}, \\
\label{rrefclause:V-dom} & \bigvee_{u \in [k]} \neg S_u(i,j) \lor \bigvee_{\ell \in [n]} V(i,j,\ell) &  i \! \in \! [s] \! \setminus \! \{1\}, j\! \in \! [t], \\
\label{rrefclause:I-dom} & \bigvee_{u\in [k]} \neg S_u(1,j) \lor \bigvee_{m \in [r]} I(j,m) & j\! \in \! [t], \\
\label{rrefclause:L-dom} & \bigvee_{u\in [k]} \neg S_u(i,j) \lor \bigvee_{j' \in [t]} L(i,j,j')  &  i \! \in \! [s] \! \setminus \! \{1\}, j\! \in \! [t], \\
\label{rrefclause:R-dom} & \bigvee_{u\in [k]} \neg S_u(i,j) \lor \bigvee_{j' \in [t]} R(i,j,j')  &  i \! \in \! [s] \! \setminus \! \{1\}, j\! \in \! [t], \\
\label{rrefclause:V-func} & \bigvee_{u\in [k]} \neg S_u(i,j) \lor \neg V(i,j,\ell) \lor \neg V(i,j,\ell')  &   i \! \in \! [s] \! \setminus \! \{1\}, j\! \in \! [t], \ell,\ell' \! \in \! [n], \ell \neq \ell',\\
\label{rrefclause:I-func} & \bigvee_{u\in [k]} \neg S_u(i,j) \lor \neg I(j,m) \lor \neg I(j,m')  &  j\! \in \! [t], m,m' \! \in \! [r], m \neq m',\\
\label{rrefclause:L-func} & \bigvee_{u\in [k]} \neg S_u(i,j) \lor \neg L(i,j,j') \lor \neg L(i,j,j'')  &   i \! \in \! [s] \! \setminus \! \{1\}, j,j',j'' \! \in \! [t], j' \neq j'',\\
\label{rrefclause:R-func} & \bigvee_{u\in [k]} \neg S_u(i,j) \lor \neg R(i,j,j') \lor \neg R(i,j,j'')  &   i \! \in \! [s] \! \setminus \! \{1\}, j,j',j'' \! \in \! [t], j' \neq j'', \\
\label{rrefclause:Sst-true} & S_u(s,t) & u \in [k], \\
\label{rrefclause:L-unused-to-zero} & \bigvee_{u\in [k]} \neg S_u(i,j) \lor \neg L(i,j,j') \lor S_{u'}(i-1,j') & i \! \in \! [s] \! \setminus \! \{1\}, j,j' \in [t], u' \in [k], \\
\label{rrefclause:R-unused-to-zero} & \bigvee_{u\in [k]} \neg S_u(i,j) \lor \neg R(i,j,j') \lor S_{u'}(i-1,j') & i \! \in \! [s] \! \setminus \! \{1\}, j,j' \in [t], u' \in [k]. 
\end{alignat}
\endgroup

Clauses in \eqref{rrefclause:axioms} - \eqref{rrefclause:R-func} are just the clauses in \eqref{refclause:axioms} - \eqref{refclause:R-func} with the additional disjuncts $\bigvee_{u \in [k]} \neg S_u(i,j)$ with the corresponding $(i,j)$. Clauses \eqref{rrefclause:Sst-true} together with \eqref{rrefclause:empty-clause} make sure that $C_{s,t}$ is empty. Clauses in \eqref{rrefclause:L-unused-to-zero} and \eqref{rrefclause:R-unused-to-zero} ensure that if $C_{i-1,j'}$ is not selected then it cannot be used as a premise.

It is immediate that the partial assignment that maps $S_u(i,j)$ to 1 for all $(i,j) \in [s] \times [t]$ and all $u \in [k]$ maps $\textnormal{R}^k\textnormal{REF}^{F}_{s,t}$ to $\textnormal{REF}^{F}_{s,t}$.

We now define the formula $\textnormal{R}^k\textnormal{REF}^{n,r}_{s,t}$ by a change to $\textnormal{R}^k\textnormal{REF}^{F}_{s,t}$ analogous to the change by which we obtained $\textnormal{REF}^{n,r}_{s,t}$ from $\textnormal{REF}^{F}_{s,t}$. That is, the clauses of $\textnormal{R}^k\textnormal{REF}^{n,r}_{s,t}$ are \eqref{rrefclause:nontaut} - \eqref{rrefclause:R-unused-to-zero} of $\textnormal{R}^k\textnormal{REF}^{F}_{s,t}$ together with the following clauses (to replace \eqref{rrefclause:axioms}):
\begin{equation}
\begin{split}
\label{rrefclause-n-r:axioms}
\bigvee_{u \in [k]} \neg S_u(1,j) \lor \neg I(j,m) \lor \neg C(m,\ell,b) \lor D(1,j,\ell,b)  \\
  & j \! \in \! [t], m \! \in \! [r], \ell \! \in \! [n], b \! \in \! \{0,1\},
\end{split}
\end{equation}
saying that if clause $C_{1,j}$ is selected and is a weakening of clause $C_m$ (described by $C$-variables), then it contains each literal of $C_m$.

\begin{theorem}
\label{thm:refl-princ-Upper-bound}
The negation of the reflection principle for resolution expressed by the formula $\textnormal{SAT}^{n,r} \land \textnormal{R}^k\textnormal{REF}^{n,r}_{s,t}$ has $\text{Res}(2)$ refutations of size $O(trn^2 + tr^2 + trnk + st^2n^3 + st^2n^2k + st^2nk^2 + st^3n)$.
\end{theorem}

\begin{proof}
By induction on $i \in [s]$ we derive for each $j \in [t]$ the formula
\begin{equation}
 D_{i,j} := \bigvee_{u \in [k]}\neg S(i,j) \lor \bigvee_{\ell \in [n], b \in \{0,1\}} \left( D(i,j,\ell,b) \land T(\ell)^b \right).
\end{equation}
Then, cutting $D_{s,t}$ with \eqref{rrefclause:empty-clause} for each $\ell \in [n]$ and $b \in \{0,1\}$, followed by $k$ cuts with clauses \eqref{rrefclause:Sst-true}, yields the empty clause.

Base case: $i=1$. For each $j \in [t], m \in [r], \ell \in [n], b \in \{0,1\}$, cut \eqref{satclause:sat-lit-in-clause} with \eqref{rrefclause-n-r:axioms} to obtain $\bigvee_{u \in [k]} \neg S_u(1,j) \lor \neg I(j,m) \lor \neg T(m,\ell,b) \lor D(1,j,\ell,b)$. Applying $\land$-introduction to this and $\neg T(m,\ell,b) \lor T(\ell)^b$ (which is either
\eqref{satclause:sat-by-positive} or \eqref{satclause:sat-by-negative}) yields 
\begin{equation}
\label{eq:refl-up-bound-base-1}
\bigvee_{u \in [k]} \neg S_u(1,j) \lor \neg I(j,m) \lor \neg T(m,\ell,b) \lor \left(D(1,j,\ell,b) \land T(\ell)^b \right). 
\end{equation}
Cutting \eqref{eq:refl-up-bound-base-1} for each $\ell \in [n]$ and $b \in \{0,1\}$ with \eqref{satclause:at-least-one-lit-sat} gives $\neg I(j,m) \lor D_{1,j}$. Cutting this for each $m \in [r]$ with \eqref{rrefclause:I-dom} yields $D_{1,j}$.

Induction step: Assume we have derived $D_{i-1,j'}$ for all $j' \in [t]$. We derive $D_{i,j}$ for each $j \in [t]$. Write $P_1$ in place of $L$ and  $P_0$ in place of $R$. 

For each $\ell \in [n], b \in \{0,1\}, j' \in [t]$, cut $\bigvee_{u \in [k]} \neg S_u(i,j) \lor \neg D(i-1,j',\ell,1) \lor \neg D(i-1,j',\ell,0)$ (from \eqref{rrefclause:nontaut}) with $\bigvee_{u \in [k]} \neg S_u(i,j) \lor \neg P_{1-b}(i,j,j') \lor \neg V(i,j,\ell) \lor D(i-1,j',\ell,1-b)$ (which is from \eqref{rrefclause:res-L-cut} or \eqref{rrefclause:res-R-cut}) to obtain $\bigvee_{u \in [k]} \neg S_u(i,j) \lor \neg P_{1-b}(i,j,j') \lor \neg V(i,j,\ell) \lor \neg D(i-1,j',\ell,b)$. Cut this with $D_{i-1,j'}$ to get
\begin{equation}
\label{eq:refl-up-bound-induc-1}
\bigvee_{u \in [k]} \neg S_u(i,j) \lor \neg P_{1-b}(i,j,j') \lor \neg V(i,j,\ell) \lor \left(D_{i-1,j'} \setminus \{ D(i-1,j',\ell,b) \land T(\ell)^b \}\right).
\end{equation}
Cutting \eqref{eq:refl-up-bound-induc-1} with $T(\ell) \lor \neg T(\ell)$ yields
\begin{align}
 \begin{split}
\label{eq:refl-up-bound-induc-2}
\bigvee_{u \in [k]} \neg S(i,j) & \lor \neg P_{1-b}(i,j,j') \lor \neg V(i,j,\ell) \lor T(\ell)^{1-b} \\
& \lor \left( D_{i-1,j'} \setminus \{ D(i-1,j',\ell,0) \land \neg T(\ell), D(i-1,j',\ell,1) \land T(\ell) \} \right).
 \end{split}
\end{align}
Next, for each $\ell' \in [n] \setminus \{\ell\}$ and $b' \in \{0,1\}$, apply $\land$-introduction to  $T(\ell') \lor \neg T(\ell')$ and $\bigvee_{u \in [k]} \neg S_u(i,j) \lor \neg P_{1-b}(i,j,j') \lor \neg V(i,j,\ell) \lor \neg D(i-1,j',\ell',b') \lor D(i,j,\ell',b')$ (from \eqref{rrefclause:res-L-transf} or \eqref{rrefclause:res-R-transf}) to get
\begin{align}
 \begin{split}
\label{eq:refl-up-bound-induc-3}
 \bigvee_{u \in [k]} \neg S_u(i,j) \lor \neg P_{1-b}(i,j,j') \lor \neg V(i,j,\ell) & \lor \left( D(i,j,\ell',b') \land T(\ell')^{b'} \right) \\
 & \lor \neg D(i-1,j',\ell',b') \lor T(\ell')^{1-b'}.
 \end{split}
\end{align}
Cutting \eqref{eq:refl-up-bound-induc-3}, for each $\ell'\in [n] \setminus \{\ell\}$ and $b' \in \{0,1\}$, with \eqref{eq:refl-up-bound-induc-2} results, after a weakening, in
\begin{equation}
\label{eq:refl-up-bound-induc-4}
\bigvee_{u \in [k]} \neg S_u(i-1,j') \lor \neg P_{1-b}(i,j,j') \lor \neg V(i,j,\ell) \lor T(\ell)^{1-b} \lor D_{i,j}.
\end{equation}
Cut \eqref{eq:refl-up-bound-induc-4}, for each $u' \in [k]$, with $\bigvee_{u \in [k]} \neg S_u(i,j) \lor \neg P_{1-b}(i,j,j') \lor S_{u'}(i-1,j')$ (from \eqref{rrefclause:L-unused-to-zero} or \eqref{rrefclause:R-unused-to-zero}) to get
\begin{equation}
\label{eq:refl-up-bound-induc-5}
 \neg P_{1-b}(i,j,j') \lor \neg V(i,j,\ell) \lor T(\ell)^{1-b} \lor D_{i,j}.
\end{equation}
Recall that we have obtained \eqref{eq:refl-up-bound-induc-5} for each $\ell \in [n], b \in \{0,1\}, j' \in [t]$. Cutting \eqref{eq:refl-up-bound-induc-5}, for each $j' \in [t]$, 
with $\bigvee_{u \in [k]} \neg S_u(i,j) \lor \bigvee_{j' \in [t]} P_{1-b}(i,j,j')$ (which is from \eqref{rrefclause:L-dom} or \eqref{rrefclause:R-dom}) yields $\neg V(i,j,\ell) \lor T(\ell)^{1-b} \lor D_{i,j}$. We have derived such clause for each $\ell \in [n], b \in \{0,1\}$, so a cut on $T(\ell)$ gives $\neg V(i,j,\ell) \lor D_{i,j}$, and cutting this, for each $\ell \in [n]$, with \eqref{rrefclause:V-dom} yields $D_{i,j}$.

As for bounding the size of the refutation, the size of the base case is $O(t(rn^2 + r^2 + rnk))$,
the total size of the induction steps is $O(st(n^3t + n^2tk + ntk^2 + nt^2 ))$, 
and the size of the finish is $O(n^2 + nk)$. Altogether, this is $O(trn^2 + tr^2 + trnk + st^2n^3 + st^2n^2k + st^2nk^2 + st^3n)$.
\end{proof}

\section{The Lower Bounds}
\label{sec:res2-no-wfdp}

We need a modification of two results of Segerlind, Buss and Impagliazzo \cite{segerlind-buss-impagliazzo}. Namely, their switching lemma works with the usual notion of width of a clause, and we would like it to work with the notion of `number of pairs mentioned' in the sense of Definition \ref{def:mentioned} below. This is because our random restrictions have to respect the functional properties of the formula $\textnormal{REF}^F_{s,t}$ (expressed by clauses \eqref{refclause:V-dom} - \eqref{refclause:R-func}), and it is therefore convenient to require that they evaluate variables in groups determined by home pair. Consequently, we do not want to represent a $k$-DNF simplified by a random restriction by a standard decision tree like in \cite{segerlind-buss-impagliazzo}, as such a tree would branch exponentially in $t$, which would prevent taking union bounds over the branches of shallow trees occurring in the proof of our switching lemma. To circumvent this problem, the decision trees we construct (called decision trees over $\textnormal{REF}^F_{s,t}$) ask queries like ``What is the left premise of clause $C_{i,j}$?" rather than queries like ``Is $L(i,j,j')$ true?". This makes their branching a bit more manageable (though still exponential in the number of variables of $F$), but there is a price to pay in terms of parameters of the switching lemma (Theorem \ref{thm:switching_a-DNFs_to_trees}) and its more complicated proof, which uses certain independence properties of our random restrictions. Also, such trees no longer represent formulas over all partial assignments, but only over assignments that do not violate the functionality axioms and evaluate variables in groups determined by home pair. Accordingly, we need to adapt to our different notions of width and representation a result in \cite{segerlind-buss-impagliazzo} which says that if the lines of a $\textnormal{Res}(k)$ refutation can be strongly represented by shallow decision trees, the refutation can be converted into a resolution refutation of a small width.

Our random restrictions (Definition \ref{def:random_restriction}) will be applied to $k$-DNFs in the variables of $\textnormal{R}^k\textnormal{REF}^F_{s,t}$ and they are defined in two stages, the first of which evaluates all the $S$-variables, thereby declaring some pairs $(i,j)$ selected (when $\bigwedge_{u \in [k]}S_u(i,j)$ evaluates to 1), and in the second stage all variables with a home pair that was not selected are evaluated randomly and independently. The restricted formula is therefore in the variables of $\textnormal{REF}^F_{s,t}$, and the purpose of the switching lemma is to show that it can be represented by a shallow decision tree over $\textnormal{REF}^F_{s,t}$ with a high probability. We begin with a definition of these trees and the notion of representation.
Please recall Definition \ref{def:home-pair-and-set-to} before reading the next one.

\begin{definition}
\label{def:decision-tree-for-REF-F-s-t}
A \emph{decision tree over} $\textnormal{REF}^{F}_{s,t}$ is a rooted tree $T$ in which every internal node is labelled with a pair $(i,j) \in [s] \times [t]$. There are $2^{2n} \cdot r$ edges leaving each node labelled with $(1,j) \in \{1\} \times [t]$, and they are labelled with pairs $(C_{1,j}, m)$, where $C_{1,j}$ is a clause in variables $x_1, \ldots, x_n$, and $m \in [r]$. There are $2^{2n} \cdot nt^2$ edges leaving each node labelled with $(i,j) \in \{2, \ldots, s\} \times [t]$, and these edges are labelled with tuples $(C_{i,j}, \ell, j', j'')$, where $C_{i,j}$ is a clause in variables $x_1, \ldots, x_n$, $\ell \in [n]$, and $j',j'' \in [t]$. 
The leaves of $T$ are labelled with either 0 or 1. No pair $(i,j)$ is allowed to label two nodes on any path from the root to a leaf of $T$. For each node $v$ of $T$, the path from the root to $v$ is viewed as a partial assignment $\pi_v$ that for each edge that is on the path, leaving a node  with a label $(i,j)$, evaluates the variables of $\textnormal{REF}^{F}_{s,t}$ with home pair $(i,j)$ in the following way: If $i=1$ and the label of the edge is $(C_{1,j}, m)$, then $\pi_v$ sets $D(1,j,\cdot, \cdot)$ to $C_{1,j}$ and $I(j, \cdot)$ to $m$; otherwise $i \in [s] \setminus \{1\}$ and the label of the edge is some tuple $(C_{i,j}, \ell, j', j'')$, in which case $\pi_v$ sets $D(i,j,\cdot, \cdot)$ to $C_{i,j}$, $V(i,j,\cdot)$ to $\ell$, $L(i,j,\cdot)$ to $j'$, and $R(i,j,\cdot)$ to $j''$.
For $b \in \{0,1\}$, we let $\textnormal{Br}_b(T)$ stand for the set of paths (viewed as partial assignments) that lead from the root to a leaf labelled with $b$. 
\end{definition}

\begin{definition}
\label{def:strongly-represents-and-index-height}
Let $G$ be a DNF in the variables of $\textnormal{REF}^{F}_{s,t}$. We say that a decision tree $T$ over $\textnormal{REF}^{F}_{s,t}$ \emph{strongly represents} $G$ if for every $\pi \in \textnormal{Br}_0(T)$, for every $q \in G$, $q \! \restriction \! \pi = 0$ and for every $\pi \in \textnormal{Br}_1(T)$, there exists $q \in G$, $q \! \restriction \! \pi = 1$. The \emph{representation index-height of} $G$, $h_{\textnormal{i}}(G)$, is the minimum height of a decision tree over $\textnormal{REF}^{F}_{s,t}$ strongly representing $G$.
\end{definition}

\begin{definition}
\label{def:mentioned}
Let $\pi$ be a partial assignment to the variables of $\textnormal{REF}^{F}_{s,t}$, and let $E$ be a clause in the variables of $\textnormal{REF}^F_{s,t}$. We say that a pair $(i,j) \in [s] \times [t]$ is \emph{mentioned in} $\pi$ (resp. $E$) if it is the home pair of a variable in $\dom(\pi)$ (resp. a literal of which is in $E$).
\end{definition}

\begin{definition}
A partial assignment $\pi$ to the variables of $\textnormal{REF}^{F}_{s,t}$ is called \emph{respectful} if for each $(i,j) \in [s] \times [t]$, either $(i,j)$ is not mentioned in $\pi$, or $i \in [s]\! \setminus \! \{1\}$ and each of $D(i,j,\cdot, \cdot)$, $V(i,j, \cdot)$, $R(i,j, \cdot)$, $L(i,j, \cdot)$ is set by $\pi$, or $i=1$ and both $D(1,j,\cdot, \cdot)$ and $I(j, \cdot)$ are set by $\pi$. In other words, respectful assignments are exactly the assignments of the form $\pi_v$ where $v$ is a node of a decision tree over $\textnormal{REF}^{F}_{s,t}$.

If $T$ is a decision tree over $\textnormal{REF}^{F}_{s,t}$ and $\pi$ is a respectful partial assignment, $T \! \restriction \! \pi$ is obtained as follows: for each node $v$ of $T$ with a label $(i,j)$ that is mentioned in $\pi$, contract the edge whose label determines an assignment to the variables with home pair $(i,j)$ that is a subset of $\pi$, and delete all other edges leaving $v$ (and delete their associated subtrees).
\end{definition}

\begin{lemma}
\label{lem:decision_tree_restriction_by_respectful}
Let $T$ be a decision tree over $\textnormal{REF}^{F}_{s,t}$, let $G$ be a DNF, and let $\pi$ be a respectful partial assignment. If $T$ strongly represents $G$, then $T \! \restriction \! \pi$ strongly represents $G \! \restriction \! \pi$.
\end{lemma}

\begin{proof}
For a leaf $v$ of $T \! \restriction \! \pi$ there is a unique leaf $u$ of $T$ such that $\pi_v = \pi_u \setminus \pi$, where $\pi_u$, $\pi_v$ are defined as in Definition \ref{def:decision-tree-for-REF-F-s-t}. Moreover, $v$ has the same label as $u$, and $\pi$ and $\pi_u$ are compatible. Therefore, for a term $q \in G$ we have $q \! \restriction \! (\pi \cup \pi_u) = q \! \restriction \! (\pi \cup \pi_v) = (q \! \restriction \! \pi) \! \restriction \! \pi_v$. Also, for $b \in \{0,1\}$, if $q \! \restriction \! \pi_u = b$ then $q \! \restriction \! (\pi \cup \pi_u) = b$.
\end{proof}

In the other direction, we have the following lemma.

\begin{lemma}
\label{lem:decision_tree_composition}
Let $T$ be a decision tree over $\textnormal{REF}^{F}_{s,t}$, and let $G$ be a DNF in the variables of $\textnormal{REF}^{F}_{s,t}$. For each leaf $v$ of $T$, let $T_v$ be a decision tree that strongly represents $G \! \restriction \! \pi_v$, where $\pi_v$ is the path in $T$ from the root to $v$. Moreover, assume that each label $(i,j)$ of an internal node of $T_v$ is a home pair of a variable of $G \! \restriction \! \pi_v$. Then the tree $T'$ obtained by appending to each leaf $v$ of $T$ the tree $T_v$ strongly represents $G$.
\end{lemma}

\begin{proof}
This follows directly from the definitions.
\end{proof}

\begin{definition}
Let $C$ be a clause in the variables of $\textnormal{REF}^{F}_{s,t}$. The \emph{index-width} of $C$ is the number of  pairs $(i,j) \in [s] \times [t]$ that are mentioned in $C$. The index-width of a resolution derivation is the maximum index-width of a clause in the derivation.
\end{definition}

The following theorem is an adaptation of \cite[Theorem 5.1]{segerlind-buss-impagliazzo}. 

\begin{theorem}
\label{thm:shallow_trees_to_small_width_resolution}
Let $H$ be a CNF in the variables of $\textnormal{REF}^{F}_{s,t}$ whose every clause has index-width at most $h \geq 1$. If for some $k \geq 1$ there is a $\textnormal{Res}(k)$ refutation of $H$ such that for each line $G$ of the refutation, $h_{\textnormal{i}}(G) \leq h$, then there is a resolution refutation of $H$ together with the functionality clauses \eqref{refclause:V-dom} - \eqref{refclause:R-func} of $\textnormal{REF}^{F}_{s,t}$ such that the index-width of the refutation is at most $3h$.
\end{theorem}

\begin{proof}
Denote $\Pi$ the  $\textnormal{Res}(k)$ refutation. For a line $G$ in $\Pi$, let $T_G$ be a decision tree over $\textnormal{REF}^{F}_{s,t}$ of minimum height that strongly represents $G$. We can assume that no node of $T_G$ is labelled with a pair $(i,j)$ that is not a home pair of any variable of $G$.

For any respectful partial assignment $\pi$ let $C_{\pi}$ be the clause consisting of the following literals: $D(i,j,\ell,b)$ if and only if $\pi(D(i,j,\ell,b)) = 0$, $\neg D(i,j,\ell,b)$ if and only if $\pi(D(i,j,\ell,b)) = 1$, $\neg I(j,m)$ if and only if $\pi$ sets $I(j, \cdot)$ to $m$, $\neg V(i,j,\ell)$ if and only if $\pi$ sets $V(i,j,\cdot)$ to $\ell$, $\neg L(i,j,j')$ if and only if $\pi$ sets $L(i,j,\cdot)$ to $j'$, $\neg R(i,j,j')$ if and only if $\pi$ sets $R(i,j,\cdot)$ to $j'$. 

By induction on the lines of $\Pi$ we show that for each line $G$ of $\Pi$ and for each $\pi \in \textnormal{Br}_0(T_G)$, there is a resolution derivation $\Pi_G(\pi)$ of $C_{\pi}$ from $H$ together with the clauses \eqref{refclause:V-dom} - \eqref{refclause:R-func}, such that the index-width of $\Pi_G(\pi)$ is at most $3h$. The theorem then follows from $\{C_{\pi} : \pi \in \textnormal{Br}_0(T_{\emptyset})\} = \{C_{\emptyset}\} = \{\emptyset\}$. 

Assume that $G$ is an axiom $X \lor \neg X$. Then all the branches of $T_G$ are labelled with 1, and so $\{C_{\pi} : \pi \in \textnormal{Br}_0(T_G)\} = \emptyset$.

Next assume that $G \in H$. Let $\pi \in \textnormal{Br}_0(T_G)$. Since $G$ is a clause, the node labels of $T_G$ are exactly the pairs $(i,j)$ mentioned in $G$. Note that since $G \! \restriction \! \pi = 0$, for every $(i,j)$ each literal of a variable in $D(i,j,\cdot, \cdot)$ that is in $G$ is also in $C_{\pi}$. Suppose that $\pi$ sets $V(i,j,\cdot)$ to $\ell \in [n]$. If there is a literal in $G$ of a variable from $V(i,j,\cdot)$ such that the literal is not in $C_{\pi}$, then the literal must be $V(i,j,\ell')$ for some $\ell' \in [n]$ with $\ell' \neq \ell$. This follows from $G \! \restriction \! \pi = 0$ and $\neg V(i,j,\ell) \in C_{\pi}$. Such literals $V(i,j,\ell')$ can be removed from $G$ by resolving with the clause $\neg V(i,j,\ell) \lor \neg V(i,j,\ell')$ from \eqref{refclause:V-func}. Similarly, we remove from $G$ the literals in $G \setminus C_{\pi}$ of $I,L,R$-variables by resolving with the corresponding clauses from \eqref{refclause:I-func}, \eqref{refclause:L-func}, \eqref{refclause:R-func}, respectively. We have thus obtained a resolution derivation $\Pi_G(\pi)$ of $C_{\pi}$ from $\{G\}$ together with the clauses \eqref{refclause:V-func} - \eqref{refclause:R-func}. Because the index-width of $G$ is at most $h$, the same is true for the clauses in $\Pi_G(\pi)$.

Now assume that line $G$ in $\Pi$ is inferred from previously derived lines $G_1, \ldots, G_d$ for $d \in [2]$. By the induction hypothesis, we have for each $c \in [d]$ and for each  $\pi \in \textnormal{Br}_0(T_{G_c})$ a resolution derivation $\Pi_G(\pi)$ of $C_{\pi}$ with the required properties. First construct a decision tree $T$ as follows: if $d=1$, $T$ is $T_{G_1}$; if $d=2$, append to each branch $\pi \in \textnormal{Br}_1(T_{G_1})$ the tree $T_{G_2} \! \restriction \! \pi$. Observe that for each $\pi \in \textnormal{Br}_0(T)$ there is $c \in [d]$ and $\pi' \in \textnormal{Br}_0(T_c)$ such that $\pi' \subseteq \pi$, and  $C_{\pi}$ is a weakening of $C_{\pi'}$. Also, the index-width of $C_{\pi}$ is at most $2h$, because so is the height of $T$. 
For a node $v$ of $T$ define a partial assignment $\pi_v$ as in Definition \ref{def:decision-tree-for-REF-F-s-t}.

Let $\sigma \in \textnormal{Br}_0(T_G)$ be given. Inductively, from the leaves to the root of $T$, we show that if a node $v$ of $T$ is such that $\pi_v$ is compatible with $\sigma$, then there is a resolution derivation $\Pi_G(\pi_v, \sigma)$ of $C_{\pi_v} \lor C_{\sigma}$ from $H$ together with the clauses \eqref{refclause:V-dom} - \eqref{refclause:R-func}, such that the index-width of $\Pi_G(\pi_v, \sigma)$ is at most $3h$. When we reach the root of $T$, we will have obtained a derivation $\Pi_G(\emptyset, \sigma)$ of $C_{\sigma}$, and this is the derivation $\Pi_G(\sigma)$ we are after.

Assume that $v$ is a leaf of $T$ and $\pi_v$ is compatible with $\sigma$. Then $\pi_v \in \textnormal{Br}_0(T)$. This can be seen as follows. It is easy to check that the rules of $\textnormal{Res}(k)$ have the property, called strong soundness, that any partial assignment that satisfies all premises of a rule also satisfies the conclusion of the rule. If we had $\pi_v \in \textnormal{Br}_1(T)$, then for each $c \in [d]$, $\pi_v$ contains some $\pi_c \in \textnormal{Br}_1(T_{G_c})$, and so $G_c \! \restriction \! \pi_v = G_c \! \restriction \! \pi_c = 1$ because $T_{G_c}$ strongly represents $G_c$. By strong soundness it follows that $G \! \restriction \! \pi_v = 1$. But this means that $\pi_v$ cannot be compatible with $\sigma$, because $\sigma$ falsifies every term of $G$. So indeed $\pi_v \in \textnormal{Br}_0(T)$. Further, we have that $C_{\pi_v} \lor C_{\sigma}$ is a weakening of $C_{\pi_v}$, which in turn is a weakening of $C_{\pi'}$ for some $\pi' \in \textnormal{Br}_0(T_c)$ and some $c \in [d]$ such that that $\pi' \subseteq \pi_v$, by the construction of $T$. By the inductive hypothesis we have a resolution refutation $\Pi_G(\pi')$ of $C_{\pi'}$ with the required properties.  Because the index-width of $C_{\pi_v}$ is at most $2h$, the index-width of $C_{\pi_v} \lor C_{\sigma}$ is at most $3h$. We have thus obtained a resolution derivation $\Pi_G(\pi_v, \sigma)$ of $C_{\pi_v} \lor C_{\sigma}$ with the required properties.

Now assume that $v$ is labelled with a pair $(i,j)$ and $\pi_v$ is compatible with $\sigma$. We distinguish two cases. In the first case, assume that $(i,j)$ is mentioned in $\sigma$. Then there is a child $u$ of $v$ such that $\pi_u \setminus \pi_v \subseteq \sigma$. Also, $\pi_u$ is compatible with $\sigma$. By the induction hypothesis we therefore have a resolution refutation $\Pi_G(\pi_u,\sigma)$ of $C_{\pi_u} \lor C_{\sigma}$ with the required properties. Because $\pi_u \cup \sigma = \pi_v \cup \sigma$, we have $C_{\pi_u} \lor C_{\sigma} = C_{\pi_v} \lor C_{\sigma}$, and so we define $\Pi_G(\pi_v,\sigma)$ to be $\Pi_G(\pi_u,\sigma)$. In the second case, assume that $(i,j)$ is not mentioned in $\sigma$. Then for each child $u$ of $v$, $\pi_u$ is compatible with $\sigma$. By the induction hypothesis, for each such $u$ there is a resolution refutation $\Pi_G(\pi_u,\sigma)$ of $C_{\pi_u} \lor C_{\sigma}$ with the required properties. Notice that $C_{\pi_u} \lor C_{\sigma} = C_{\pi_u \setminus \pi_v} \lor C_{\pi_v} \lor C_{\sigma}$. We first construct a resolution refutation $\Pi'$ of $\{ C_{\pi_u \setminus \pi_v} : u \textnormal{ is a child of } v \}$ together with the clauses \eqref{refclause:V-dom} - \eqref{refclause:R-dom} such that the index-width of $\Pi'$ is 1. 
This is easy: since $\{ C_{\pi_u \setminus \pi_v} : u \textnormal{ is a child of } v \} = \{C_{\alpha} : \alpha \textnormal{ is respectful and mentions just the pair } (i,j) \}$, we use \eqref{refclause:V-dom}, \eqref{refclause:L-dom}, \eqref{refclause:R-dom} (resp. \eqref{refclause:I-dom} if $i=1$) to remove all the negated $V,L,R$-variables (resp. the negated $I$-variables) from the clauses $C_{\alpha}$, and we refute the resulting clauses by a refutation in the form of a complete binary tree to resolve all the $D$-variables.
Now, having $\Pi'$, we define $\Pi_G(\pi_v, \sigma)$ as follows: add the literals of $C_{\pi_v} \lor C_{\sigma}$ to each clause of $\Pi'$ other than an initial clause from \eqref{refclause:V-dom}, \eqref{refclause:L-dom}, \eqref{refclause:R-dom}, \eqref{refclause:I-dom}, and derive each initial clause $C_{\pi_u} \lor C_{\sigma}$ in the resulting derivation using the derivation $\Pi_G(\pi_u,\sigma)$. It is easy to see that $\Pi_G(\pi_v, \sigma)$ has the required properties. 
\end{proof}

We now turn our attention to the formula $\textnormal{R}^k\textnormal{REF}^{F}_{s,t}$. Recall from its definition in Section \ref{sec:refl-princ-upper-bound} that its variables are those of  $\textnormal{REF}^{F}_{s,t}$ together with variables $S_u(i,j)$, $(i,j) \in [s] \times [t]$, $u \in [k]$.
In the following definition we extend the notion of home pair from Definition \ref{def:home-pair-and-set-to} to the $S$-variables, and we extend the notion of a pair being mentioned accordingly.

\begin{definition}
For $(i,j) \in [s] \times [t]$ and $u \in [k]$, the \emph{home pair} of the variable $S_u(i,j)$ is $(i,j)$.

We say that a pair $(i,j)$ is \emph{mentioned in} a clause $E$ (resp. a partial assignment $\pi$; a term $q$) if it is a home pair of a variable a literal of which is in $E$ (resp. which is in $\dom(\pi)$; a literal of which is in $q$).
\end{definition}

\begin{definition}
Let $U \subseteq [s] \times [t]$ and let $G$ be a DNF in the variables of $\textnormal{R}^k\textnormal{REF}^{F}_{s,t}$. If for each term $q \in G$ there is $(i,j) \in U$ such that $(i,j)$ is mentioned in $q$, then we say that $U$ is an \emph{index-cover} of $G$. The \emph{index-covering number of} $G$, $c_{\textnormal{i}}(G)$, is the minimum cardinality of an index-cover of $G$.
\end{definition}

\begin{definition}
For a set $U \subseteq [s] \times [t]$, denote by $\text{Var}(U)$ the set of all variables of $\textnormal{REF}^{F}_{s,t}$ with home pair in $U$, that is, 
\[
 \text{Var}(U) := \bigcup_{(i,j) \in U} D(i,j,\cdot, \cdot) \cup \bigcup_{(i,j) \in U \setminus ([1] \times [t])} \left( R(i,j,\cdot) \cup L(i,j,\cdot) \cup V(i,j,\cdot) \right) \cup \bigcup_{(1,j) \in U} I(j,\cdot).
\]
Also, denote by $\text{Var}_S(U)$ the set of all $S$-variables with home pair in $U$; in symbols, $\text{Var}_S(U) := \{ S_u(i,j) : u \in [k], (i,j) \in U \}$.
\end{definition}

We generalize random restrictions from \cite{atserias-muller2019-focs} to our case of $\textnormal{R}^k\textnormal{REF}^{F}_{s,t}$.

\begin{definition}
\label{def:random_restriction}
A \emph{random restriction} $\rho_k$ is a partial assignment to the variables of $\textnormal{R}^k\textnormal{REF}^{F}_{s,t}$ given by the following experiment:
\begin{enumerate}
\item Independently for each $(i,j) \in [s] \times [t]$ and $u \in [k]$, map $S_u(i,j)$ to 0 or 1, each with probability $1/2$. 
\item Let $A$ be the set of those $(i,j) \in [s] \times [t]$ such that for every $u \in [k]$, $S_u(i,j)$ is mapped to 1.
\item Map independently each variable from $\text{Var}(([s] \times [t]) \setminus A)$ to 0 or 1, each with probability $1/2$.
\end{enumerate}
\end{definition}

\begin{theorem}
\label{thm:switching_a-DNFs_to_trees}
Suppose that $k \geq 1,a \geq 1$ are integers such that $k \geq a$. There is $\delta > 0$ and an integer $n_0 > 0$ such that if $n,r,s,t$ are integers satisfying 
\begin{equation}
\label{eq:conditions_n,r,t_in_switching_theorem} 
r \leq t \leq 2^{\delta n} \text{ and } n_0 \leq n,
\end{equation} 
and $F$ is a CNF with $r$ clauses in $n$ variables, then for every $a$-DNF $G$ in the variables of $\textnormal{R}^k\textnormal{REF}^{F}_{s,t}$ and every $w>0$, 
\begin{equation}
\label{eqn:probability_p_a_w}
\Pr[h_{\textnormal{i}}(G \! \restriction \! \rho_k) > w] \leq 2^{- \frac{w}{n^{a-1}} \gamma(a)},
\end{equation}
where $\gamma(a) = \frac{(\log e)^a}{2^{a^2 + 3a -2} a!}$.
\end{theorem}

\begin{proof}
Denote the right hand side of the inequality \eqref{eqn:probability_p_a_w} by $p_a(w)$. Let $k \geq 1$ be given and denote $\rho := \rho_k$. We prove the theorem by induction on $a$.

Base case: $a = 1$. $G$ is a clause. If $c_{\text{i}}(G) \leq w$, then $\Pr[h_{\textnormal{i}}(G \! \restriction \! \rho) > w] = 0$ because we can build a decision tree strongly representing $G \! \restriction \! \rho$ by querying the pairs from the smallest index-cover of $G$. If $c_{\text{i}}(G) > w$, we have $\Pr[h_{\textnormal{i}}(G \! \restriction \! \rho) > w] \leq \Pr[G \! \restriction \! \rho \neq 1] \leq \left( 1 - (1-2^{-k})/2 \right)^{c_{\text{i}}(G)} \leq \left( 1 - 1/4 \right)^{c_{\text{i}}(G)}  \leq e^{- c_{\text{i}}(G)/4} = 2^{- c_{\text{i}}(G) \gamma(1)} \leq 2^{-w \gamma(1)}$.

Induction step: Assume the theorem holds for $a-1$, witnessed by $\delta(k,a-1)$ and $n_0(k,a-1)$. Find a positive $\delta(k,a) \leq \delta(k,a-1)$ and an integer $n_0(k,a) \geq n_0(k,a-1)$ such that
\begin{equation}
\label{eq:condition_on_delta_and_n0}
 - \frac{\gamma(a-1)}{2}n + \left(2 \log t + \log n + \frac{\gamma(a-1)}{n^{a-2}} \right) \cdot \frac{\gamma(a-1)}{4} \leq - \gamma(a)
\end{equation}
holds for any $n,r,t$ satisfying \ref{eq:conditions_n,r,t_in_switching_theorem} with $\delta(k,a)$ and $n_0(k,a)$ in place of $\delta$ and $n_0$, respectively. Let $G$ be an $a$-DNF, and let $U$ be an index cover of $G$ of size $c_{\textnormal{i}}(G)$. We distinguish two cases based on $c_{\textnormal{i}}(G)$.

Case 1: $c_{\textnormal{i}}(G) > \frac{w}{n^{a-1}} \cdot \frac{\gamma(a-1)}{4}$. In this case we want to show that $\rho$ satisfies $G$ with a high probability. To this end, note that there are at least $c_{\text{i}}(G)/a$ many terms in $G$ that are index-independent, that is, for no two of them there is a pair $(i,j) \in [s] \times [t]$ mentioned by both. (If every such set of terms was smaller than $c_{\text{i}}(G)/a$, take a maximal one and observe that the set of pairs mentioned by the terms forms an index-cover of $G$ of cardinality smaller than $c_{\text{i}}(G)$, a contradiction.) It is easy to see that each of these index-independent terms is satisfied by $\rho$ with independent probability at least $2^{-2a}$. Therefore,
\begin{align*}
  \Pr[h_{\text{i}}(G \! \restriction \! \rho) > w] 
& \leq \Pr[G \! \restriction \! \rho \neq 1] \leq \left( 1 - 2^{-2a} \right)^{c_{\text{i}}(G)/a} 
  \leq 2^{-\frac{(\log e)}{a 2^{2a}} c_{\text{i}}(G)} 
  \leq 2^{-\frac{(\log e)}{a 2^{2a}} \cdot \frac{w}{n^{a-1}} \cdot \frac{\gamma(a-1)}{4}} \\
& = 2^{ - \frac{w}{n^{a-1}}\gamma(a)}.
\end{align*}
This finishes the inductive step for Case 1. 

Case 2: $c_{\textnormal{i}}(G) \leq \frac{w}{n^{a-1}} \cdot \frac{\gamma(a-1)}{4}$. 
Let $U' \subseteq U$, and let $\nu : \text{Var}_S(U) \cup \text{Var}(U \setminus U') \rightarrow \{0,1\}$ satisfy the following conditions:
\begin{enumerate}[($\nu$1)]
\item \label{item:nu_1_on_U'} for each $(i,j) \in U'$ and each $u \in [k]$, $\nu(S_u(i,j)) = 1$, 
\item \label{item:nu_0_outside_U'} for each $(i,j) \in U \setminus U'$ there is $u \in [k]$ with $\nu(S_u(i,j)) = 0$.
\end{enumerate}  
We have
\begin{align*}
& \Pr[h_{\text{i}}(G \! \restriction \! \rho) > w \mid \rho \! \restriction \! \dom(\nu) = \nu]  \\
& \leq \Pr [ \exists \pi : \text{Var}(U') \rightarrow \{0,1\}, \pi \text{ is respectful} \land h_{\text{i}}((G \! \restriction \! \pi) \! \restriction \! \rho) > w - |U'| \mid \rho \! \restriction \! \dom(\nu) = \nu] \\
& \leq   \sum_{\substack{\pi : \text{Var}(U') \rightarrow \{0,1\}, \\ \pi \text{ is respectful}}} \Pr [h_{\text{i}}((G \! \restriction \! \pi) \! \restriction \! \rho) > w- |U'| \mid \rho \! \restriction \! \dom(\nu) = \nu] \\
& = \sum_{\substack{\pi : \text{Var}(U') \rightarrow \{0,1\}, \\ \pi \text{ is respectful}}} \Pr [h_{\text{i}}(((G \! \restriction \! \pi) \! \restriction \!  \nu) \! \restriction \! \rho) > w- |U'|] \\
& \leq \left(t^2 n 2^{2n} \right)^{|U'|} p_{a-1}(w - |U'|).
\end{align*}
Here the first inequality follows from Lemma \ref{lem:decision_tree_composition} and from $(G \! \restriction \! \pi) \! \restriction \! \rho = (G \! \restriction \! \rho) \! \restriction \! \pi$ (since $\dom(\pi) \cap \dom(\rho) = \emptyset$). The second inequality is obtained by the union bound. The equality follows since the events $h_{\text{i}}(((G \! \restriction \! \pi) \! \restriction \! \nu) \! \restriction \! \rho) > w- |U'|$ and $\rho \! \restriction \! \dom(\nu) = \nu$ are independent (by the definition of $\rho$). And the last inequality is by the induction hypothesis and by the upper bound $t^2n2^{2n} = \max \{t^2n2^{2n}, r2^{2n} \}$ (recall that $t \geq r$) over $(i,j) \in [s] \times [t]$ on the number of respectful partial assignments mentioning exactly the pair $(i,j)$.  

Since the event $A \cap U = U'$ (where the random variable $A$ is given by Definition \ref{def:random_restriction}) is the disjoint union of events $\rho \! \restriction \! \dom(\nu) = \nu$ over all $\nu$ satisfying conditions \ref{item:nu_1_on_U'} and \ref{item:nu_0_outside_U'}, the above calculation implies 
\begin{equation}
\label{eq:prob_representing_tree>w_conditioned_on_A_intersect_U_is_U'}
 \Pr[h_{\text{i}}(G \! \restriction \! \rho) > w \mid A \cap U = U'] \leq  \left(t^2 n 2^{2n} \right)^{|U'|} p_{a-1}(w - |U'|).
\end{equation}
Therefore,
\begin{align}
 \Pr[h_{\text{i}}(G \! \restriction \! \rho) > w] & = \sum_{U' \subseteq U} \Pr[h_{\text{i}}(G \! \restriction \! \rho) > w \land A \cap U = U'] \nonumber \\
& = \sum_{U' \subseteq U} \Pr[h_{\text{i}}(G \! \restriction \! \rho) > w \mid A \cap U = U'] \cdot \Pr[A \cap U = U'] \nonumber \\
& \leq \sum_{U' \subseteq U}  \left(t^2 n 2^{2n} \right)^{|U'|} p_{a-1}(w - |U'|) \cdot 2^{-k |U'|} \left( 1 - 2^{-k} \right)^{|U \setminus U'|}  \nonumber \\
& = \sum_{q = 0}^{c_{\text{i}}(G)}  \binom{c_{\text{i}}(G)}{q} \left(t^2 n 2^{2n} \right)^{q} p_{a-1}(w - q) \cdot 2^{-kq} \left( 1 - 2^{-k} \right)^{c_{\text{i}}(G) - q}  \nonumber \\
& \leq \left(t^2 n 2^{2n} \right)^{c_{\text{i}}(G)} p_{a-1}(w - c_{\text{i}}(G)). \label{eq:prob_representing_tree>w_interm_step}
\end{align}
Here the first inequality is by \ref{eq:prob_representing_tree>w_conditioned_on_A_intersect_U_is_U'} and by the definition of $\rho$. The second inequality follows from $\left(t^2 n 2^{2n} \right)^{q} p_{a-1}(w - q) \leq \left(t^2 n 2^{2n} \right)^{c_{\text{i}}(G)} p_{a-1}(w - c_{\text{i}}(G))$ for $q \leq c_{\text{i}}(G)$. From \ref{eq:prob_representing_tree>w_interm_step}, using the definition of $p_{a-1}(w - c_{\text{i}}(G))$ and the assumption $c_{\text{i}}(G) \leq \frac{w}{n^{a-1}} \cdot \frac{\gamma(a-1)}{4}$, we get
\begin{align*}
\log (\Pr[h_{\text{i}}(G \! \restriction \! \rho) > w]) & \leq \left( 2 \log t + \log n + 2n \right) c_{\text{i}}(G) - \frac{w - c_{\text{i}}(G)}{n^{a-2}} \gamma(a-1) \\
& = \left(2 \log t + \log n + 2n + \frac{\gamma(a-1)}{n^{a-2}} \right) c_{\text{i}}(G) - \frac{w \gamma(a-1)}{n^{a-2}} \\
& \leq \left(2 \log t + \log n + 2n + \frac{\gamma(a-1)}{n^{a-2}} \right) \frac{w}{n^{a-1}} \cdot \frac{\gamma(a-1)}{4} - \frac{w \gamma(a-1)}{n^{a-2}} \\
& = - \frac{w \gamma(a-1)}{2n^{a-2}} + \left(2 \log t + \log n + \frac{\gamma(a-1)}{n^{a-2}} \right) \frac{w}{n^{a-1}} \cdot \frac{\gamma(a-1)}{4} \\
& \leq  - \frac{w}{n^{a-1}} \gamma(a),
\end{align*}
where the last inequality is equivalent to \ref{eq:condition_on_delta_and_n0}. This finishes the inductive step for Case 2, and the proof of the theorem.
\end{proof}

We now show an index-width lower bound on resolution refutations of $\textnormal{REF}^{F}_{s,t}$ for an unsatisfiable $F$. This was done in \cite{atserias-muller2019-focs} for a non-layered version of the formula, of which our $\textnormal{REF}^{F}_{s,t}$ is a restriction, so the index-width lower bound we need does not immediately follow from that in \cite{atserias-muller2019-focs}. We provide a simpler proof for $\textnormal{REF}^{F}_{s,t}$. First a definition. 

\begin{definition} \label{def:admissible}
A partial assignment $\sigma$ to the variables of $\textnormal{REF}^{F}_{s,t}$ is called \emph{admissible} if it satisfies all the following conditions.
\begin{enumerate}[({A}1)]
\item \label{item:admissible-set-or-untouched} For each $(i,j) \in [s] \times [t]$, $D(i,j,\cdot,\cdot)$ (resp. $V(i,j,\cdot)$, $I(j,\cdot)$, $L(i,j,\cdot)$, $R(i,j,\cdot)$) either is set to some clause (resp. some $\ell \in [n]$, some $m \in [r]$, some $j' \in [t]$, some $j' \in [t]$) by $\sigma$ or contains no variable that is in $\dom(\sigma)$. 

\item \label{item:admissible-RL-implies-D} For each $(i,j) \in [s] \times [t]$, if $L(i,j,\cdot)$ or $R(i,j,\cdot)$ is set to some  $j' \in [t]$, then both $D(i,j,\cdot,\cdot)$ and $D(i-1,j',\cdot,\cdot)$ are set.

\item \label{item:admissible-D-iff-V,I} For each $(i,j) \in ([s] \setminus \{1\}) \times [t]$, $D(i,j,\cdot,\cdot)$ is set if and only if $V(i,j,\cdot)$ is set. For each $j \in [t]$,  $D(1,j,\cdot,\cdot)$ is set if and only if $I(j,\cdot)$ is set.

\item \label{item:admissible-D-nontaut-and-fat} For each $(i,j) \in [s] \times [t]$, if $D(i,j,\cdot,\cdot)$ is set to a clause $C_{i,j}$, then $C_{i,j}$ is non-tautological and has at least $\min\{s-i, n\}$ many literals. If $D(i,j,\cdot,\cdot)$ is set to a clause $C_{i,j}$ with less than $n$ literals and $V(i,j,\cdot)$ is set to some $\ell \in [n]$, then none of the literals of $x_\ell$ is in $C_{i,j}$.

\item \label{item:admissible-empty-clause} If $D(s,t, \cdot , \cdot)$ is set, it is set to the empty clause.

\item \label{item:admissible-I}For each $j \in [t]$, if $I(j,\cdot)$ is set, then $\sigma$ satisfies all clauses in \eqref{refclause:axioms} with this $j$.

\item \label{item:admissible-RL-cut-variable_and_RL-transf} For each $(i, j) \in ([s] \setminus \{1\}) \times [t]$, if $L(i,j,\cdot )$ (resp. $R(i,j,\cdot )$) is set, then $\sigma$ satisfies all clauses in \eqref{refclause:res-L-cut} and \eqref{refclause:res-L-transf} (resp. \eqref{refclause:res-R-cut} and \eqref{refclause:res-R-transf}) with this $(i,j)$ (i.e., those clauses that contain the literal $\neg L(i,j,j')$ (resp. $\neg R(i,j,j')$) for some $j' \in [t]$).
\end{enumerate}
\end{definition}

\begin{theorem}
\label{thm:width_lb}
Let $w>0$. If $n,r,s,t$ are integers satisfying 
\begin{equation}
\label{eq:conditions_on_n,r,s,t,w_in_width_lb_theorem} 
2 \leq n+1 \leq s, \quad 2w < t,
\end{equation}
and $F$ is an unsatisfiable CNF consisting of $r$ clauses $C_1,\ldots, C_r$ in $n$ variables $x_1,\ldots, x_n$, then any resolution refutation of $\textnormal{REF}^{F}_{s,t}$ has index-width greater than $w$.
\end{theorem}
 
\begin{proof}
Assume for a contradiction that there is a resolution refutation $\Pi$ of $\textnormal{REF}^{F}_{s,t}$ of index-width at most $w$. We will show that if there is an admissible partial assignment falsifying a clause $E$ in $\Pi$ obtained by the resolution rule from $E_0$ and $E_1$, then there is an admissible partial assignment falsifying either $E_0$ or $E_1$. This immediately (by induction) leads to a contradiction, since the empty assignment is admissible and falsifies the last (empty) clause in $\Pi$, and, by definition, no partial admissible assignment falsifies any clause of $\textnormal{REF}^{F}_{s,t}$.

Let then $\sigma$ be an admissible partial assignment falsifying a clause $E$ in $\Pi$. Without loss of generality, assume that $\sigma$ is a minimal (with respect to inclusion) admissible partial assignment with this property. 

Let $Q$ be the variable resolved on to obtain $E$ from $E_0$ and $E_1$. If $Q \in \dom(\sigma)$, then $\sigma$ already falsifies either $E_0$ or $E_1$. So assume that $Q \not \in \dom(\sigma)$.
We consider two cases.

Case 1. Suppose that for some $(i,j) \in [s] \times [t]$, $Q \in D(i,j,\cdot, \cdot)$ or $Q \in V(i,j,\cdot)$ (resp. $Q \in I(j,\cdot)$ and $i = 1$). Note that by \ref{item:admissible-set-or-untouched}, \ref{item:admissible-RL-implies-D}, and \ref{item:admissible-D-iff-V,I}, no variable from $D(i,j,\cdot, \cdot) \cup V(i,j,\cdot) \cup L(i,j,\cdot) \cup R(i,j,\cdot)$ (resp. $D(1,j,\cdot, \cdot) \cup I(j,\cdot)$) is in $\dom(\sigma)$, and, moreoever, for any $j' \in [t]$, it is not the case that $L(i+1,j',\cdot)$ or $R(i+1,j',\cdot)$ is set to $j$ by $\sigma$. Therefore, we can extend $\sigma$ to a partial assignment $\sigma'$ as follows. Set $D(i,j,\cdot, \cdot)$ to any non-tautological clause containing $n$ literals, unless $(i,j) = (s,t)$, in which case set $D(i,j,\cdot, \cdot)$ to the empty clause. In case $i \geq 2$, set $V(i,j,\cdot)$ to an arbitrary value $\ell \in [n]$; in case $i=1$, set $I(j,\cdot)$ to any $m \in [r]$ such that the clause $C_m$ is a subset of the clause to which we have set $D(1,j,\cdot, \cdot)$. (Here we use that $F$ is unsatisfiable.) It is straightforward to check that $\sigma'$ is admissible. Since $Q \in \dom(\sigma')$, $\sigma'$ falsifies $E \cup \{Q^{1-\sigma'(Q)}\}$, of which either $E_0$ or $E_1$ is a subset.   

Case 2. Suppose that for some $(i,j) \in ([s] \setminus \{1\}) \times [t]$, $Q \in L(i,j, \cdot)$ (if $Q \in R(i,j, \cdot)$, we proceed in a completely analogous way). We may assume that $D(i,j,\cdot, \cdot)$ is set to some clause $C_{i,j}$ by $\sigma$ and $V(i,j,\cdot)$ is set to some $\ell \in [n]$ by $\sigma$; if not, set them both as described in Case 1. We now concentrate on the level $i-1$. Since the index-width of $E$ is at most $w$ and $\sigma$ is a minimal admissible partial assignment falsifying $E$, 
\begin{equation}
\label{eq:bad_indices_on_previous_level_upper_bound}
 \left| \{j': D(i-1,j', \cdot, \cdot) \text{ is set by }  \sigma  \} \right| \leq 2w.
\end{equation}
This is because $D(i-1,j', \cdot, \cdot)$ can be set by $\sigma$ for two reasons: either $(i-1,j')$ is mentioned in $E$ (which, together with \ref{item:admissible-RL-implies-D} and \ref{item:admissible-D-iff-V,I}, implies that $D(i-1,j', \cdot, \cdot)$ is set by $\sigma$) or there is some $j'' \in [t]$ such that a literal of a variable from $L(i,j'', \cdot)$ or $R(i,j'', \cdot)$ is in $E$ (which forces $\sigma$ to set $L(i,j'', \cdot)$ or $R(i,j'', \cdot)$, respectively, in order to falsify the literal) and $\sigma$ happens to set $L(i,j'', \cdot)$ or $R(i,j'', \cdot)$, respectively, to $j'$ (and therefore by \ref{item:admissible-RL-implies-D} $D(i-1,j', \cdot, \cdot)$ must be set by $\sigma$ too).

We extend $\sigma$ to a partial assignment $\sigma'$ as follows. Set $L(i,j, \cdot)$ to any $j'$ that is not from the set in \eqref{eq:bad_indices_on_previous_level_upper_bound}. Such $j'$ exists because $2w < t$. Thanks to that, set $D(i-1,j',\cdot, \cdot)$ to the clause $C_{i-1,j'} := (C_{i,j} \setminus \{\neg x_{\ell}\} ) \cup \{x_{\ell}\}$, where $C_{i,j}$ and $\ell$ are as above. Finally, if $i \in \{3,\ldots, s\}$, then either $C_{i-1,j'}$ has less than $n$ literals and we set $V(i-1,j',\cdot)$ to any $\ell' \in [n]$ such that no literal of $x_{\ell'}$ is in $C_{i-1,j'}$, or $C_{i-1,j'}$ has $n$ literals, in which case we set $V(i-1,j',\cdot)$ arbitrarily. If $i=2$, then by \ref{item:admissible-D-nontaut-and-fat}, \eqref{eq:conditions_on_n,r,s,t,w_in_width_lb_theorem}, and the definition of $C_{i-1,j'}$ we know that $C_{i-1,j'}$ has $n$ literals, and we set $I(j',\cdot)$ to any $m \in [r]$ such that $C_m \subseteq C_{i-1,j'}$. (Here we use that $F$ is unsatisfiable.) This finishes the definition of $\sigma'$. 

It is again easy to verify that $\sigma'$ is admissible. Because $Q \in \dom(\sigma')$, $\sigma'$ falsifies $E \cup \{Q^{1-\sigma'(Q)}\}$, of which one of $E_0$, $E_1$ is a subset. 
\end{proof}	

We now put together all the results so far in this section to show a length lower bound on $\textnormal{Res}(k)$ refutations of $\textnormal{R}^k\textnormal{REF}^{F}_{s,t}$ with an unsatisfiable $F$. 

\begin{theorem}
\label{thm:main_size_lb_for_RkREF^F_st}
Suppose $k \geq 1$ is an integer. There is $\delta > 0$ and an integer $n_0 > 0$ such that if $n,r,s,t$ are integers satisfying 
\begin{equation}
\label{eq:conditions_on_n,r,s,t_in_the_main_theorem} 
n_0 \leq n, \quad n+1 \leq s \leq t, \quad r \leq t \leq 2^{\delta n}, \quad n^k \leq t,
\end{equation}
and $F$ is an unsatisfiable CNF consisting of $r$ clauses $C_1,\ldots, C_r$ in $n$ variables $x_1,\ldots, x_n$, then any $\text{Res}(k)$ refutation of $\textnormal{R}^k\textnormal{REF}^{F}_{s,t}$ has length greater than $2^{\beta(k) \frac{t}{n^{k-1}}}$, where $\beta(k) := \frac{(\log e)^k}{2^{k^2 + 4k + 4} k!}$. 
\end{theorem}

\begin{proof}
Let $k \geq 1$ be given. Take $\delta$ and $n_0$ as given by Theorem \ref{thm:switching_a-DNFs_to_trees} for $a=k$. If necessary, increase $n_0$ so that it satisfies 
\begin{equation}
\label{eq:condition_on_n_0_in_main_thm}
\beta(k) n_0 > k + 1.
\end{equation}
Let $n,r,s,t$ be integers satisfying \eqref{eq:conditions_on_n,r,s,t_in_the_main_theorem}, and let $F$ satisfy the hypothesis of the theorem. 
Assume for a contradiction that there is a $\text{Res}(k)$ refutation $\Pi$ of $\textnormal{R}^k\textnormal{REF}^{F}_{s,t}$ of length at most $2^{\beta(k) \frac{t}{n^{k-1}}}$. 

Recall the random variable $A$ from Definition \ref{def:random_restriction}. We have that with probability  $2^{-k}$, 
\begin{enumerate}[(a)]
\item \label{item:s,t_in_A-in_main_theorem} $(s,t) \in A$.
\end{enumerate}

By the Chernoff bound and the union bound, with probability at least $1 - se^{- t2^{-k} / 8} $, 
\begin{enumerate}[(b)]
\item \label{item:A_widely_intersects_levels-in_main_thereom} for each $i \in [s]$ the cardinality of $A \cap (\{i\} \times [t])$ is at least $t/2^{k+1}$. 
\end{enumerate}
We have 
\begin{equation*}
se^{- t2^{-k} / 8} = 2^{\log s - \frac{t \log e}{2^{k+3}}} \leq 2^{\log n_0 - \frac{n_0 \log e}{2^{k+3}}}  < 2^{-(k+1)},
\end{equation*}
where we used $s \leq t$, $n_0 \leq s$ (from \eqref{eq:conditions_on_n,r,s,t_in_the_main_theorem}), and \eqref{eq:condition_on_n_0_in_main_thm}.

By Theorem \ref{thm:switching_a-DNFs_to_trees} and the union bound, with probability at least $1 - |\Pi| \cdot 2^{- \frac{t}{n^{k-1}2^{k+5}} \gamma(k)}$, 
\begin{enumerate}[(c)]
\item \label{item:all_lines_to_trees-in_main_theorem} for every line $G$ in $\Pi$, $h_{\text{i}}(G \! \restriction \! \rho_k) \leq t/2^{k+5}$. 
\end{enumerate}
We have 
\begin{equation*}
 |\Pi| \cdot 2^{- \frac{t}{n^{k-1}2^{k+5}}\gamma(k)} 
 \leq 2^{\beta(k) \frac{t}{n^{k-1}}} \cdot 2^{- \frac{t}{n^{k-1}2^{k+5}} \gamma(k)} 
 = 2^{ - \beta(k) \frac{t}{n^{k-1}}} 
\leq 2^{ - \beta(k) n_0} < 2^{-(k+1)},
\end{equation*}
where we used $n^k \leq t$, $n_0 \leq n$ (from \eqref{eq:conditions_on_n,r,s,t_in_the_main_theorem}), and \eqref{eq:condition_on_n_0_in_main_thm}.

It follows that there exists $\rho_k$ such that \ref{item:s,t_in_A-in_main_theorem}, \ref{item:A_widely_intersects_levels-in_main_thereom} and \ref{item:all_lines_to_trees-in_main_theorem} hold. Fix any such $\rho_k$ and denote it by $\rho$. We now restrict $\textnormal{R}^k\textnormal{REF}^{F}_{s,t} \! \restriction \! \rho$ some more before we apply Theorem \ref{thm:shallow_trees_to_small_width_resolution}.

For each level $i \in [s]$ select any $t' := \lfloor t/2^{k+1} \rfloor - 2$ home pairs $(i,j)$ of variables of $\textnormal{R}^k\textnormal{REF}^{F}_{s,t} \! \restriction \! \rho$ (they exist thanks to  \ref{item:A_widely_intersects_levels-in_main_thereom}), making sure to include the pair $(s,t)$ in the selection. Denote the set of selected pairs by $B$. Define a partial assignment $\nu: \text{Var}(\textnormal{R}^k\textnormal{REF}^{F}_{s,t} \! \restriction \! \rho) \rightarrow \{0,1\}$ by mapping all the variables with not selected home pairs so that they form an arbitrary resolution derivation from $F$, that is, so that $\nu$ satisfies every clause of $\textnormal{R}^k\textnormal{REF}^{F}_{s,t} \! \restriction \! \rho$ that contains a literal of a variable in $\dom(\nu)$. (This derivation may require two clauses per level, which is why we selected only $\lfloor t/2^{k+1} \rfloor - 2$ on each level.) Note that $\nu$ is respectful. Hence by \ref{item:all_lines_to_trees-in_main_theorem} and Lemma \ref{lem:decision_tree_restriction_by_respectful} we have that for any line $G$ in $\Pi \! \restriction \! \rho$, $h_{\text{i}}(G \! \restriction \! \nu) \leq t/2^{k+5}$.

Next, define a partial assignment $\lambda$ as follows. For every $(i,j) \in B \setminus (\{1\} \times [t])$ and every $j' \in [t]$ such that $(i-1,j') \not \in B$, map both $L(i,j,j')$ and $R(i,j,j')$ to 0. 
Let us verify that $((\textnormal{R}^k\textnormal{REF}^{F}_{s,t} \! \restriction \! \rho) \! \restriction \! \nu) \! \restriction \! \lambda$ is $\textnormal{REF}^{F}_{s,t'}$ up to a re-indexing of variables determined by a bijection that maps, for each $i \in [s]$, the elements of $B \cap (\{i\} \times [t])$ to $(i,1),\ldots, (i,t')$. Thanks to \ref{item:s,t_in_A-in_main_theorem}, clauses \eqref{rrefclause:Sst-true} are satisfied by $\rho$. All clauses \eqref{rrefclause:L-unused-to-zero} and \eqref{rrefclause:R-unused-to-zero} are satisfied: if $(i,j) \in B$ and $(i-1,j') \notin B$, then the clause is satisfied by $\lambda$, otherwise it is satisfied by $\rho$ or $\nu$. Clauses \eqref{rrefclause:axioms} - \eqref{rrefclause:R-func} with $(i,j) \notin B$ are satisfied either by $\rho$ (if $(i,j) \notin A$) or by $\nu$. Clauses \eqref{rrefclause:axioms} - \eqref{rrefclause:R-func} with $(i,j) \in B$ become, after removing those clauses \eqref{rrefclause:res-L-cut} - \eqref{rrefclause:res-R-transf} that are satisfied by $\lambda$ and after the re-indexing of variables, the clauses \eqref{refclause:axioms} - \eqref{refclause:R-func} with $t$ replaced by $t'$. (Here notice that clauses \eqref{rrefclause:empty-clause} become \eqref{refclause:empty-clause} thanks to $(s,t) \in B$.) Hence $((\textnormal{R}^k\textnormal{REF}^{F}_{s,t} \! \restriction \! \rho) \! \restriction \! \nu) \! \restriction \! \lambda$ is indeed $\textnormal{REF}^{F}_{s,t'}$ up to the re-indexing of variables. 

Let us now show that for a line $G$ in $(\Pi \! \restriction \! \rho) \! \restriction \! \nu$ we have that $G \! \restriction \! \lambda$ is, after the re-indexing of variables, strongly represented by a decision tree over $\text{REF}^F_{s,t'}$ of height at most $t/2^{k+5}$. As we already verified, $h_{\text{i}}(G) \leq t/2^{k+5}$, and therefore there is a tree $T$ over $\text{REF}^F_{s,t}$ of minimum height which strongly represents $G$ and whose height is at most $t/2^{k+5}$. Define a tree $T \! \restriction \! \lambda$ by deleting all edges (and the corresponding subtrees) in $T$ whose label is of the form $(C_{i,j}, \ell, j', j'')$ with $(i-1,j') \notin B$ or $(i-1,j'') \notin B$. $T \! \restriction \! \lambda$ is, after relabelling its nodes and edges according to the re-indexing bijection, a decision tree over $\textnormal{REF}^{F}_{s,t'}$. With every branch $\pi$ of $T \! \restriction \! \lambda$ we associate a partial assignment $\pi_{T \restriction \lambda}: \text{Var}(((\textnormal{R}^k\textnormal{REF}^{F}_{s,t} \! \restriction \! \rho) \! \restriction \! \nu) \! \restriction \! \lambda) \rightarrow \{0,1\}$ defined via the re-indexing bijection and Definition \ref{def:decision-tree-for-REF-F-s-t}, understanding the relabelled $T \! \restriction \! \lambda$ as a tree over $\textnormal{REF}^{F}_{s,t'}$. But every branch $\pi$ of $T \! \restriction \! \lambda$ is also a branch of $T$, hence Definition \ref{def:decision-tree-for-REF-F-s-t} with $T$ (which is a tree over $\textnormal{REF}^{F}_{s,t}$) says how $\pi$ should be viewed as a partial assignment to $\text{Var}(\textnormal{REF}^{F}_{s,t})$; let us denote the partial assignment by $\pi_T$ for clarity. It is easy to see from the definitions that for every branch $\pi$ in $T \! \restriction \! \lambda$, $\dom(\lambda) \cap \dom(\pi_{T \restriction \lambda}) = \emptyset$ and $\pi_T \subseteq \lambda \cup \pi_{T \restriction \lambda}$. It follows that $G \! \restriction \! \lambda$ is strongly represented by $T \! \restriction \! \lambda$. The tree $T \! \restriction \! \lambda$ has, of course, height at most $t/2^{k+5}$. 

We can now apply Theorem \ref{thm:shallow_trees_to_small_width_resolution} taking $\text{REF}^F_{s,t'}$ (i.e., the re-indexed $((\textnormal{R}^k\textnormal{REF}^{F}_{s,t} \! \restriction \! \rho) \! \restriction \! \nu) \! \restriction \! \lambda$) for $H$, $t'$ for $t$, and $t/2^{k+5}$ for $h$, to obtain a resolution refutation of $\text{REF}^F_{s,t'}$ of index-width at most $3t/2^{k+5}$. 

But we have 
\begin{equation*}
2 \cdot 3t/2^{k+5} < t/2^{k+2} <  \lfloor t/2^{k+1} \rfloor - 2 = t',
\end{equation*}
where the second inequality follows from \ref{eq:conditions_on_n,r,s,t_in_the_main_theorem} and \ref{eq:condition_on_n_0_in_main_thm}. Therefore, we can use Theorem \ref{thm:width_lb}, taking $3t/2^{k+5}$ for $w$ and $t'$ for $t$, to conclude that any resolution refutation of $\text{REF}^F_{s,t'}$ has index-width greater than $3t/2^{k+5}$. That is a contradiction.
\end{proof}

\section{Proofs of Theorems \ref{thm:res(k)-no-wfdp}
 and \ref{thm:res(k)-not_automatable}}

\begin{proof}[Proof of Theorem \ref{thm:res(k)-no-wfdp}]
Denote by $F$ the well-known CNF $\neg \textnormal{PHP}^{n+1}_n$ called the negation of the pigeonhole principle, expressing that a multi-valued function from $n+1$ to $n$ is injective. It consists of $r := n+1+ (n^3+n^2)/2$ clauses in $\widetilde{n} := (n+1)n$ variables. 

Define $A_n := \textnormal{SAT}^{\widetilde{n},r} \! \restriction \! \gamma_F$, where $\gamma_F$ is as in Proposition \ref{propos:subst-to-SAT}. 

Since by \cite{Krajicek-Pudlak-Woods1995, Pitassi-Beame-Impagliazzo1993} there exists $\alpha > 0$ and an integer $n_1$ such that for every $n \geq n_1$, $\neg \textnormal{PHP}^{n+1}_n$ has no $\text{Res}(k)$ refutations of size at most $2^{n^{\alpha}}$, the same is true for $A_n$. This is because by Proposition \ref{propos:subst-to-SAT} there is a substitution $\tau$ such that $A_n \! \restriction \! \tau$ is $\neg \textnormal{PHP}^{n+1}_n$ together with some tautological clauses, and if $\Pi$ is a $\text{Res}(k)$ refutation of $A_n$ then $\Pi \! \restriction \! \tau$ is a $\text{Res}(k)$ refutation of $A_n \! \restriction \! \tau$. This shows item \ref{item:lower_bound_on_A_n}.

Define $B_{n,k} := \textnormal{R}^k\textnormal{REF}^{F}_{s,t}$, where we set $s := \widetilde{n} +1$ and $t := \widetilde{n}^k$.

Let $\delta > 0$ and integer $n_0$ witness Theorem \ref{thm:main_size_lb_for_RkREF^F_st}. Set $n_2 \geq n_0$ so that the hypotheses \eqref{eq:conditions_on_n,r,s,t_in_the_main_theorem} with $\widetilde{n}$ in place of $n$ hold with our choice of $r,s,t$ (as functions of $\widetilde{n}$) for all $\widetilde{n} \geq n_2$. 
By that theorem, for every $\widetilde{n} \geq n_2$, any $\textnormal{Res}(k)$ refutation of $B_{n,k}$ has size greater than $2^{\beta(k) \widetilde{n}}$. Item \ref{item:lower_bound_on_B_n} follows.

Note that $\textnormal{R}^k\textnormal{REF}^{\widetilde{n},r}_{s,t} \! \restriction \!  \gamma_F$ is $\textnormal{R}^k\textnormal{REF}^{F}_{s,t}$, because $\gamma_F$ turns the clauses \eqref{rrefclause-n-r:axioms} into \eqref{rrefclause:axioms} (and the clauses satisfied by $\gamma_F$ are removed). By Theorem \ref{thm:refl-princ-Upper-bound} there is a $\textnormal{Res}(2)$ refutation of $\textnormal{SAT}^{\widetilde{n},r} \land \textnormal{R}^k\textnormal{REF}^{\widetilde{n},r}_{s,t}$ of size $O(k^2 \widetilde{n}^{3k+3})$. Hence the same holds true for $A_n \land B_{n,k} = \textnormal{SAT}^{\widetilde{n},r} \! \restriction \! \gamma_F \land \textnormal{R}^k\textnormal{REF}^{\widetilde{n},r}_{s,t} \! \restriction \!  \gamma_F$. This gives item \ref{item:upper_bound_on_A_n_and_B_n}.
\end{proof}

Theorem \ref{thm:res(k)-not_automatable} follows immediately from the more general Theorem \ref{thm:res(k)-not_automatable-more_general} below. A function $T: \mathbb{N} \rightarrow \mathbb{N}$ is called \emph{time-constructible} if there is an algorithm that when given $1^n$ (the string of $n$ many 1's) computes $1^{T(n)}$ in time $O(T(n))$. We call a function $T: \mathbb{N} \rightarrow \mathbb{N}$ \emph{subexponential} if $T(n) \leq 2^{n^{o(1)}}$. 

\begin{theorem}
\label{thm:res(k)-not_automatable-more_general}
Let $T: \mathbb{N} \rightarrow \mathbb{N}$ be time-constructible, non-decreasing and subexponential. 
If there is an integer $k \geq 1$ such that $\textnormal{Res}(k)$ is automatable in time $T$, then there are $c_1, c_2, c_3, c_4 > 0$ and an algorithm that when given as input a 3-CNF $F$ in $n$ variables decides in time $c_3 (T(c_1 n^{c_2 k}) + n^{k})^{c_4}$ whether $F$ is satisfiable.
\end{theorem}

\begin{proof}
Assume that for some integer $k \geq 1$ the system $\textnormal{Res}(k)$ is automatable in time $T$ satisfying the assumptions of the theorem. Set $r, s$ and $t$ as functions of $n$ as follows: $r := \binom{2n}{3}$, $s := n+1$, $t := n^{k+3}$. 

By Theorem \ref{thm:refl-princ-Upper-bound} there are integers $c_1, c_2 > 0$ such that $\textnormal{SAT}^{n,r} \land \textnormal{R}^k\textnormal{REF}^{n,r}_{s,t}$ has a $\textnormal{Res}(2)$ refutation $\Pi$ of size at most $c_1 n^{c_2 k}$; if necessary, increase $c_1$ and $c_2$ so that the size of $\Pi$ plus the size of the formula $\textnormal{R}^k\textnormal{REF}^{n,r}_{s,t}$ is at most $c_1 n^{c_2 k}$. 

Let $\delta > 0$ and integer $n_0 > 0$ witness Theorem \ref{thm:main_size_lb_for_RkREF^F_st}.  
Let $n_1 > n_0$ be such that for all $n \geq n_1$,
\begin{equation}
\label{eqn:cond_on_t_in_proof_of_non-automatability}
r \leq t \leq 2^{\delta n} 
\end{equation}
and 
\begin{equation}
\label{eqn:l.b._on_refutations_greater_than_runtime_of_automating}
2^{\beta(k)\frac{t}{n^{k-1}}} > T(c_1 n^{c_2 k}),
\end{equation}
where $\beta(k)$ is as in Theorem \ref{thm:main_size_lb_for_RkREF^F_st}. Here we use that $T$ is subexponential.

Define algorithm $M$ as follows. Given as input a 3-CNF $F$ in $n$ variables, check if $n \geq n_1$. If $n < n_1$, use brute force to decide if $F$ is satisfiable or not, and output the answer. If $n \geq n_1$, compute the formula $\textnormal{R}^k\textnormal{REF}^{F}_{s,t}$ and run the automating algorithm on this formula for up to $T(c_1 n^{c_2 k})$ steps. If the automating algorithm returns a $\textnormal{Res}(k)$ refutation of $\textnormal{R}^k\textnormal{REF}^{F}_{s,t}$, then output `satisfiable'. Else output `unsatisfiable'. 

Since both computing $\textnormal{R}^k\textnormal{REF}^{F}_{s,t}$ from $F$ and checking whether the output of the automating algorithm is a $\textnormal{Res}(k)$ refutation of $\textnormal{R}^k\textnormal{REF}^{F}_{s,t}$ are polynomial-time procedures, and since $T$ is time-constructible, it follows that there are $c_3, c_4 > 0$ such that the running time of $M$ is at most $c_3 (T(c_1 n^{c_2 k}) + n^{k})^{c_4}$. 
It suffices to show that $M$ gives the correct answer on 3-CNFs $F$ in $n \geq n_1$ variables such that each clause of $F$ has exactly three literals. Let $F$ be such a 3-CNF, and let $r'$ be the number of its clauses. We have $r' \leq r = \binom{2n}{3}$.

Assume first that $F$ is satisfiable. Let $\gamma_F$ and $\tau$ be as in Proposition \ref{propos:subst-to-SAT}, and let $\nu$ be a satisfying assignment for $F$. We have
\begin{equation*}
(((\textnormal{SAT}^{n,r'} \land \textnormal{R}^k\textnormal{REF}^{n,r'}_{s,t}) \! \restriction \!  \gamma_F) \! \restriction \!  \tau) \! \restriction \!  \nu 
= ((\textnormal{SAT}^{n,r'} \! \restriction \! \gamma_F) \! \restriction \!  \tau) \! \restriction \!  \nu  \land \textnormal{R}^k\textnormal{REF}^{n,r'}_{s,t} \! \restriction \!  \gamma_F 
= \textnormal{R}^k\textnormal{REF}^{F}_{s,t},
\end{equation*}
because by Proposition \ref{propos:subst-to-SAT}, $( \textnormal{SAT}^{n,r'} \! \restriction \! \gamma_F ) \! \restriction \! \tau$ is $F$ together with some tautological clauses in the variables $x_1,\ldots, x_n$. 
Let $\Pi'$ be the $\textnormal{Res}(2)$ refutation of $\textnormal{SAT}^{n,r'} \land \textnormal{R}^k\textnormal{REF}^{n,r'}_{s,t}$ given by Theorem \ref{thm:refl-princ-Upper-bound}. 
Then $\Pi'' := ((\Pi' \! \restriction \!  \gamma_F) \! \restriction \!  \tau) \! \restriction \!  \nu$ is a $\textnormal{Res}(2)$ refutation of $\textnormal{R}^k\textnormal{REF}^{F}_{s,t}$ (note that it is actually a resolution refutation), and we have
\begin{align*}
\textnormal{size}(\Pi'') + \textnormal{size}(\textnormal{R}^k\textnormal{REF}^{F}_{s,t}) 
& \leq \textnormal{size}(\Pi') + \textnormal{size}(\textnormal{R}^k\textnormal{REF}^{n,r'}_{s,t}) \\
& \leq \textnormal{size}(\Pi) + \textnormal{size}(\textnormal{R}^k\textnormal{REF}^{n,r}_{s,t}) \\
& \leq c_1 n ^{c_2 k}.
\end{align*}
Because $T$ is non-decreasing, the automating algorithm finds within the allotted time $T(c_1 n^{c_2 k})$ a $\textnormal{Res}(k)$ refutation of $\textnormal{R}^k\textnormal{REF}^{F}_{s,t}$, and $M$ outputs `satisfiable'. 

Assume now that $F$ is unsatisfiable. From our choices of $r,s,t$ and $n_1$ and from \eqref{eqn:cond_on_t_in_proof_of_non-automatability} it follows that the hypotheses \eqref{eq:conditions_on_n,r,s,t_in_the_main_theorem} 
 of Theorem \ref{thm:main_size_lb_for_RkREF^F_st} are met for all $n \geq n_1$, and the same is true with $r'$ in place of $r$. By that theorem, any $\text{Res}(k)$ refutation of $\textnormal{R}^k\textnormal{REF}^{F}_{s,t}$ has size greater than $2^{\beta(k) \frac{t}{n^{k-1}}}$. Thanks to \eqref{eqn:l.b._on_refutations_greater_than_runtime_of_automating} this implies that the automating algorithm cannot output any $\text{Res}(k)$ refutation of $\textnormal{R}^k\textnormal{REF}^{F}_{s,t}$ within the allotted time. $M$ therefore outputs `unsatisfiable'. 
\end{proof}

\section{Conclusion}
We have shown that for every integer $k \geq 2$, the system $\textnormal{Res}(k)$ does not have the weak feasible disjunction property and, unless P = NP, it is not automatable. Because of the factor $t/n^{k-1}$ that appears in the exponent of the lower bound in Theorem \ref{thm:main_size_lb_for_RkREF^F_st} and originates in the switching lemma (Theorem \ref{thm:switching_a-DNFs_to_trees}), we have not been able to extend the results to better than barely superconstant $k$. A more important open question is to rule out weak automatability of these systems assuming some standard hardness assumption.

\paragraph{Acknowledgement.} I am grateful to Albert Atserias and Jan Kraj\'{i}\v{c}ek for their valuable comments on an earlier version of the paper. I would like to thank Ilario Bonacina and Moritz M\"uller for several related conversations.

\bibliography{mybiblio}

\end{document}